\documentclass[sigplan,10pt,nonacm]{acmart}
\setcopyright{none}

\usepackage{booktabs}   

\usepackage{subcaption} 

\usepackage{graphicx}
\usepackage{listings}

\usepackage{float}
\usepackage{multicol}
\usepackage{nicefrac}

\usepackage{rotating}
\usepackage{amsfonts}
\usepackage{amssymb}
\usepackage{amsmath}
\newcommand\norm[1]{\left\lVert#1\right\rVert}
\usepackage{enumerate}
\usepackage{xcolor}
\usepackage{array}

\usepackage{amsthm}
\usepackage{stmaryrd}
\usepackage{graphbox}
\usepackage[symbol]{footmisc}

\usepackage{paralist}
\usepackage{mathrsfs}
\usepackage{pdfrender}
\newcommand*{\yes}{%
	\textpdfrender{
		TextRenderingMode=FillStroke,
		LineWidth=0.75pt, 
	}{\checkmark}%
}
\newcommand*{\no}{%
	\textpdfrender{
		TextRenderingMode=FillStroke,
		LineWidth=1pt, 
	}{\times}%
}
\usepackage{makecell}
\newcolumntype{C}[1]{>{\centering\let\newline\\\arraybackslash\hspace{0pt}}m{#1}}
\newcolumntype{L}[1]{>{\raggedright\let\newline\\\arraybackslash\hspace{0pt}}p{#1}}
\usepackage{hhline}
\newcommand{\pvars}{\mathbf{V}}
\newcommand{\fnames}{\mathbf{F}}
\newcommand{\monomials}{\mathbf{M}}
\newcommand{\vals}[1]{\mathbb{R}^{#1}}
\newcommand{\ret}{\texttt{ret}}
\newcommand{\labels}{\mathbf{L}}
\renewcommand{\paragraph}[1]{\smallskip\noindent\textsc{#1.}}
\newcommand{\satisfies}{\models}
\renewcommand{\mid}{\vert}
\newcommand{\obj}[1]{\langle \mathit{#1} \rangle}
\newcommand{\kw}[1]{\mbox{`\textbf{#1}'}}
\newcommand{\mainFunction}{{f_\texttt{main}}}
\newcommand{\loc}{\ell}
\newcommand{\lin}{\loc_{\texttt{in}}}
\newcommand{\lout}{\loc_{\texttt{out}}}
\newcommand{\locs}{\labels}
\newcommand{\precond}{\mathsf{Pre}}
\newcommand{\postcond}{\mathsf{Post}}
\newcommand{\invariant}{\mathsf{Inv}}
\newcommand{\inv}{\invariant}
\newcommand{\ind}{\mathsf{Ind}}
\newtheorem*{theorem*}{Theorem}
\newcommand{\strongInvAlgo}{\mathsf{StrongInvSynth}}
\newcommand{\weakInvAlgo}{\mathsf{WeakInvSynth}}
\newcommand{\recursiveStrongInvAlgo}{\mathsf{RecStrongInvSynth}}
\newcommand{\recursiveWeakInvAlgo}{\mathsf{RecWeakInvSynth}}

\usepackage{listing}
\usepackage{listings}
\usepackage{multirow}
\usepackage{url}
\lstdefinelanguage{prog}
{
	morekeywords={prob, if, then, else, fi, while, do, od, true, false, and, or, skip, return},
	sensitive = false
}
\newcommand{\func}[1]{\textsf{#1}}
\usepackage{amsthm}
\usepackage{amsmath}
\usepackage{mathtools}

\newtheorem{remark}{Remark}
\newtheorem*{proof}{Proof}
\theoremstyle{definition}
\begin{document}
	
	\title[Polynomial Invariant Generation]{Polynomial Invariant Generation for Non-deterministic Recursive Programs} 
	
	\author{Krishnendu Chatterjee}
	\affiliation{
		\institution{IST Austria}            
		\city{Klosterneuburg}
		\country{Austria}                    
	}
	\email{krishnendu.chatterjee@ist.ac.at}          
	
	\author{Hongfei Fu}
	\affiliation{
		\institution{Shanghai Jiao Tong University}            
		\city{Shanghai}
		\country{China}                    
	}
	\email{fuhf@cs.sjtu.edu.cn}    
	
	\author{Amir Kafshdar Goharshady}
	\affiliation{
		\institution{IST Austria}            
		\city{Klosterneuburg}
		\country{Austria}                    
	}
	\email{amir.goharshady@ist.ac.at}          
	
	\author{Ehsan Kafshdar Goharshady}
	\affiliation{
		\institution{Ferdowsi University of Mashhad}            
		\city{Mashhad}
		\country{Iran}                    
	}
	\email{e.goharshady1@gmail.com}          
	
	\begin{abstract}

We consider the classical problem of invariant generation for programs with polynomial assignments and focus on synthesizing invariants that are a conjunction of strict polynomial \emph{inequalities}. We present a sound and semi-complete method based on positivstellens\"{a}tze, i.e.~theorems in semi-algebraic geometry that characterize positive polynomials over a semi-algebraic set.

On the theoretical side, the worst-case complexity of our approach is subexponential, whereas the worst-case 
complexity of the previous complete method (Kapur, ACA 2004) is doubly-exponential.  Even when restricted to linear invariants, the best previous complexity for complete invariant generation is exponential (Col\'{o}n et al, CAV~2003). On the practical side, we reduce the invariant generation problem to quadratic programming (QCLP), which is a classical optimization problem with many industrial solvers. We demonstrate the applicability of our approach by providing experimental results on several academic benchmarks. To the best of our knowledge, the only previous invariant generation method
 that provides completeness guarantees for invariants consisting of polynomial inequalities is~(Kapur, ACA 2004), which relies on quantifier elimination and cannot even handle toy programs such as our running example.
\end{abstract}


	
	\maketitle

	\section{Introduction} \label{sec:intro}

\paragraph{Invariants} An assertion at a program location that is always satisfied by the variables whenever the location is reached is called an {\em invariant}. Invariants are essential for many quantitative analyses, as well as for fundamental problems such as proving termination~\cite{halbwachs1997verification,henzinger1994model,ngo2018bounded,cav17}. Invariant generation is a classical problem in verification and programming languages, and has been studied for decades, e.g.~for safety and liveness analysis~\cite{DBLP:books/daglib/0080029,cousot1978automatic,cousot1977abstract}.

\paragraph{Inductive Invariants} An {\em inductive} assertion is an assertion that holds at a location for the first visit to it and is preserved under every cyclic execution path to and from the location. Inductive assertions are guaranteed to be invariants, and the well-established method to prove an assertion is an invariant is to find an inductive invariant that strengthens it~\cite{sriram,DBLP:books/daglib/0080029}. 

\paragraph{Abstract Interpretation} One technique to find inductive invariants is {\em abstract interpretation}~\cite{cousot1978automatic}, which is primarily a theory of semantic approximations. It has been used for invariant generation by computing least fixed points of abstractions of the collecting semantics, but it guarantees completeness only for rare special cases~\cite{giacobazzi1997completeness}.

\paragraph{Linear vs Polynomial Invariants} For \emph{linear} invariant generation over programs with \emph{linear} updates, a sound and complete methodology was obtained by~\cite{sriram}. We consider programs with \emph{polynomial} updates and the problem of generating \emph{polynomial} invariants, i.e.~invariants that are a conjunction of polynomial \emph{inequalities} over program variables. Hence, our setting is more general than~\cite{sriram} in terms of the programs we analyze, and also the desired invariants. The only previous approach that provides completeness for this problem is~\cite{kapur2004automatically}. However, it has doubly-exponential complexity and is not practically applicable even to toy programs. Conversely, efficient but incomplete methods were proposed in~\cite{DBLP:journals/pacmpl/KincaidCBR18,DBLP:conf/fmcad/FarzanK15,DBLP:conf/pldi/KincaidBBR17}. Polynomial invariants are more desirable than linear invariants for a variety of reasons. First, there are many cases, such as the benchmarks in~\cite{benchmarks} and programs in reinforcement learning~\cite{qlearn}, where linear assertions are not enough and verification goals require higher-degree polynomial inequalities. Second, even when the desired assertions are linear, they might not be provable by means of a linear \emph{inductive} invariant, i.e.~the inductivity might require non-linearity. Finally, many programs have polynomial assignments and guards. For such programs, even when looking for linear inductive invariants, our approach is the first applicable method with completeness guarantees.

\paragraph{Motivation for Polynomial Invariants} Given that polynomial invariants provide greater expressiveness in comparison with linear invariants, they improve solutions to many classical problems, such as the following:
\setdefaultleftmargin{3mm}{0em}{0em}{0em}{0em}{0em}
\begin{compactitem}
	\item \emph{Safety Verification.} This is one of the most well-studied model checking problems: Given a program and a set of safety assertions that must hold at specific points of the program, prove that the assertions hold or report that they might be violated by the program. Many existing approaches for safety verification rely on invariants to prove the desired assertions (see~\cite{DBLP:books/daglib/0080029,alur2006predicate,padon2016ivy,albarghouthi2012ufo}). In these cases, weak invariants can lead to an increase in false positives, i.e.~if the supplied invariants are inaccurate and grossly overestimate the program's behavior, then the verifier might falsely infer that a true assertion can be violated. 
	\item \emph{Termination Analysis.} A principal approach in proving termination of programs is to synthesize ranking functions~\cite{floyd1993assigning}. Virtually all synthesis algorithms for ranking functions depend on invariants, e.g.~\cite{coloon2001synthesis,bradley2005linear,chen2007discoveringictac}. Having inaccurate invariants, such as linear instead of polynomial, can lead to a failure in the synthesis and hence inability to prove termination. The same point also applies to termination analysis of probabilistic programs~\cite{chakarov2013probabilistic,wang2019cost,huang2019modular}.
	\item \emph{Inferring Complexity Bounds.} Another fundamental problem is to find automated algorithms that infer asymptotic complexity bounds on the runtime of (recursive) programs. Current algorithms for tackling this problem, such as~\cite{cav17}, rely heavily on invariants and their accuracy. Inaccurate invariants can lead to an over-approximation of complexity or even failure to synthesize any complexity bound.
\end{compactitem}
These points not only justify the use of polynomial invariants, but also the need for completeness guarantees. Previous state-of-the-art approaches in polynomial invariant generation either lack such guarantees or have doubly-exponential runtime and cannot be applied even to toy programs.

\paragraph{Our Contribution} We consider two variants of the invariant generation problem. Informally, the \emph{weak} variant asks for an optimal invariant w.r.t.~a given objective function, while the \emph{strong} variant asks for a representative set of all invariants. Our contributions are as follows:
\begin{compactitem}
	\item \emph{Soundness and Semi-completeness.} We present a sound and semi-complete method to generate polynomial invariants for programs with polynomial updates. Our completeness requires a compactness condition that is satisfied by all real-world programs (Remark~\ref{rem:compact}). We also show that, using the standard notions of pre and post-conditions, our method can be extended to handle recursion as well.
	\item \emph{Theoretical Complexity.} We show that the worst-case complexity of our procedure is \emph{subexponential} if we consider polynomial invariants with rational coefficients. In comparison, complexity of the procedure in~\cite{kapur2004automatically} is doubly-exponential and the approach of~\cite{sriram}, which is sound and complete for \emph{linear} invariants, has \emph{exponential} complexity, whereas we show how to generate \emph{polynomial} invariants in \emph{subexponential} time. 
	\item \emph{Practical Approach.} We present a polynomial-time reduction from weak invariant generation to quadratic programming (QCLP). Solving QCLPs is an active area of research in optimization and there are many industrial solvers for handling its real-world instances. Using our algorithm, practical improvements to such solvers carry over to polynomial invariant generation. 
\end{compactitem}
Hence, our main contribution is theoretical, i.e.~presenting a sub-exponential sound and \emph{semi-complete} method for generating polynomial invariants. Moreover, we also demonstrate the applicability of our approach by providing experimental results on several academic examples from~\cite{benchmarks} that require polynomial invariants. Unsurprisingly, we observe that our approach is slower than previous sound but incomplete methods, so there is a trade-off between completeness and efficiency. However, we expect practical improvements in solving QCLPs to narrow the efficiency gap in the future. On the other hand, the only previous complete method, proposed in~\cite{kapur2004automatically}, is extremely impractical and cannot handle any of our benchmarks, not even our toy running example.

\paragraph{Techniques} While the approaches of~\cite{sriram,kapur2004automatically} use Farkas' lemma and quantifier elimination to generate invariants, our technique is based on a positivstellensatz. Our method replaces the quantifier elimination step with either (i)~an algorithm of~\cite{grigoriev1988solving} for characterizing solutions of systems of polynomial inequalities or (ii)~a reduction to QCLP.

\subsection{Related works}
 Automated invariant generation has received much attention in the past years, and various classes of approaches have been proposed, including recurrence analysis~\cite{DBLP:journals/pacmpl/KincaidCBR18,DBLP:conf/fmcad/FarzanK15,HumenbergerJK17,DBLP:conf/pldi/KincaidBBR17}, abstract interpretation~\cite{DBLP:conf/sas/BagnaraRZ05,DBLP:conf/sas/ChakarovS14,DBLP:journals/scp/Rodriguez-CarbonellK07,DBLP:conf/esop/CousotCFMMMR05,DBLP:journals/ipl/Muller-OlmS04}, constraint solving~\cite{kapur2004automatically,DBLP:conf/sas/KatoenMMM10,chencav15,DBLP:conf/atva/FengZJZX17,sriram,DBLP:conf/popl/SankaranarayananSM04,rodriguez2004automatic,DBLP:conf/atva/OliveiraBP16,DBLP:conf/popl/ChatterjeeNZ17,DBLP:journals/fcsc/YangZZX10,DBLP:conf/vmcai/Cousot05,DBLP:journals/fcsc/LinWYZ14}, inference~\cite{DBLP:conf/vmcai/GulwaniSV09,DBLP:conf/oopsla/DilligDLM13,DBLP:journals/fmsd/SharmaA16}, interpolation~\cite{DBLP:conf/tacas/McMillan08}, symbolic execution~\cite{DBLP:conf/icse/CsallnerTS08}, dynamic analysis~\cite{DBLP:conf/icse/NguyenKWF12} and learning~\cite{DBLP:conf/popl/0001NMR16}.

\begin{table*}
	\resizebox{\linewidth}{!}{
		
			  \setlength\extrarowheight{0pt}
			\begin{tabular}{|c|c|c|c|c|c|c|c|c|c|}
				\hline
				\textbf{Approach} & \makecell{\textbf{Assignments} \textbf{and} \textbf{Guards}} & \textbf{Invariants} & \textbf{Nondet} & \textbf{Rec} & \textbf{Prob}  & \textbf{Sound}  & \textbf{Complete} & \textbf{Weak}& \textbf{Strong}\\
				\hline
				\hline
				This Work & Polynomial & Polynomial & $\yes$ & $\yes$ & $\no$ & $\yes$ & $\yes^\blacklozenge$  & \makecell{$\yes$\\ QCLP} & \makecell{$\yes$ \\ Subexp}\\
				\hline
				
				\makecell{\cite{sriram}\\ CAV'03} & Linear$^c$ & Linear & $\yes$ & $\no$ & $\no$ & $\yes$ & $\yes$ & \makecell{$\yes$ \\ Exp$^\dagger$} & \makecell{$\yes$\\ Exp$^\dagger$}\\
				\hline
				
				\makecell{\cite{kapur2004automatically} \\ ACA'04} & Polynomial & Polynomial & $\yes$ & $\yes$ & $\no$ & $\yes$ & $\yes$ & \makecell{$\yes$\\ 2Exp} & \makecell{$\yes$\\2Exp}\\ \hline
				
				\makecell{\cite{dill}\\OOPSLA'13} & General & Linear (Presburger)& $\yes$ & $\yes$ & $\no$ & $\yes$ & $\no$ & $\no$ & $\no$\\ \hline
				
				\makecell{\cite{DBLP:conf/atva/FengZJZX17}\\ ATVA'17} & Polynomial & Polynomial & $\no$ & $\no$ & $\yes$ & $\yes$ & $\yes^a$ & \makecell{$\yes$ \\ Poly} & $\no$\\
				\hline
				
				\makecell{\cite{joel}\\LICS'18} & Linear$^\ddagger$ & Polynomial \emph{Equalities} & $\yes$ & $\no$ & $\no$ & $\yes^\ddagger$ & $\yes^\ddagger$ & $\no$ & $\yes^{\ddagger,b}$ \\
				\hline
				
				\makecell{\cite{DBLP:journals/pacmpl/KincaidCBR18}\\ POPL'18} & Polynomial, Exponential, Logarithmic & Polynomial,  Exponential, Logarithmic & $\yes$ & $\yes$ & $\no$ & $\yes$ & $\no$ & $\no$ & $\no$ \\
				\hline
				
				\makecell{\cite{rodriguez2004automatic}* \\ ISSAC'04} & Polynomial, Exponential & Polynomial \emph{Equalities} & $\yes$ & $\no$ & $\no$ & $\yes$ & $\yes$ & \makecell{$\yes^b$} & \makecell{$\yes^b$} \\
				\hline
				
				\makecell{\cite{DBLP:conf/popl/SankaranarayananSM04}\\ POPL'04} & Polynomial$^c$ & Polynomial \emph{Equalities} & $\yes$ & $\no$ & $\no$ & $\yes$ & $\yes^b$ & \makecell{$\yes^b$} & \makecell{$\yes^b$} \\
				\hline
				
				\makecell{\cite{DBLP:conf/fmcad/FarzanK15}\\FMCAD'15} & \makecell{General$^d$} & \makecell{General$^d$} & $\yes$ & $\no^e$ & $\no$ & $\yes$ & $\no$ & $\no$ & $\no$\\
				\hline 
				
				\makecell{\cite{DBLP:conf/pldi/KincaidBBR17}\\PLDI'17} & General & General & $\yes$ & $\yes$ & $\no$ & $\yes$ & $\no$ & $\no$ & $\no$\\
				\hline
				
				\makecell{\cite{DBLP:conf/atva/OliveiraBP16}\\ATVA'16} & Polynomial, \textit{Without Conditional Branching} & Polynomial \emph{Equalities} & $\yes$ & $\no$ & $\no$ & $\yes$ & $\yes$ & \makecell{$\yes$\\Poly} & \makecell{$\yes$\\Poly} \\
				\hline
				
				\makecell{\cite{HumenbergerJK17}* \\ ISSAC'17} & Polynomial$^\ddagger$ & Polynomial \emph{Equalities} & $\yes$ & $\no$ & $\no$ & $\yes$ & $\yes^\ddagger$ & $\yes^{\ddagger, b}$ & $\yes^{\ddagger, b}$ \\
				\hline
				
				\makecell{\cite{adje2015property}$^*$ \\ SAS'15} & Polynomial & Polynomial & $\no$ & $\no$ & $\no$ & $\yes$ & $\no$ & $\no$ & $\no$\\ \hline
			\end{tabular}
	}

	\begin{flushleft}
		\begin{footnotesize}
			\begin{multicols}{2}
			$^\blacklozenge$ Semi-complete, assuming compactness (see Remark~\ref{rem:compact} and Lemma~\ref{lem:complete})\\
			$^\dagger$ Generates a system of quadratic inequalities, but then applies quantifier elimination, leading to exponential runtime. \\
			$^\ddagger$ Treats branching conditions as non-determinism.\\
			$*$ Does not support nested loops.\\ \columnbreak
			$^a$ Semi-complete\\
			$^b$ Uses Gr\"{o}bner basis computations (super-exponential in worst-case).\\
			$^c$ Considers general transition systems instead of programs.\\
			$^d$ Handles non-linearity using linearization heuristics.\\
			$^e$ Can be extended to handle recursion (see~\cite{DBLP:conf/pldi/KincaidBBR17}).\\
			\end{multicols}
		\end{footnotesize}
	\end{flushleft}
	\caption{Summary of approaches for invariant generation.}
	\label{tab:comp}
\end{table*}

\paragraph{Summary} A summary of the results of the literature w.r.t.~types of assignments, type of generated invariants, programming language features that can be handled (i.e.~non-determinism, probability and recursion), soundness, completeness, and whether the approach can handle weak/strong invariant generation is presented in Table~\ref{tab:comp}. For approaches that are applicable to weak/strong invariant generation, the respective runtimes are also reported. Most previous methods are indeed incomparable with our approach, because they handle different problems, e.g.~different types of programs. We present the first applicable sound and semi-complete approach for polynomial invariant generation. Our complexity (subexponential) is not only better than the previous doubly-exponential complexity for polynomial invariants~\cite{kapur2004automatically}, it even beats the exponential complexity of complete methods for linear invariants~\cite{sriram}.

\paragraph{Recurrence Analysis} While approaches based on recurrence analysis can derive exact invariants, they are applicable to a restricted class of programs where closed-form solutions exist. Our approach does not require closed-form solutions.

\paragraph{Abstract Interpretation} This is the oldest and most classical approach to invariant generation~\cite{cousot1978automatic,cousot1977abstract} and has also been used for generating quadratic invariants~\cite{adje2010coupling}. However, unlike our approach, it cannot provide completeness, except in very special cases~\cite{giacobazzi1997completeness}. There are efficient tools and algorithms for invariant generation using abstract interpretation~\cite{vechev1,vechev2}, but they focus on generating linear invariants.

\paragraph{Constraint Solving} Our approach falls in this category. First, we handle polynomial invariants, thus extending approaches based on linear arithmetics, such as~\cite{DBLP:conf/sas/KatoenMMM10,sriram,DBLP:conf/atva/OliveiraBP16,DBLP:conf/popl/ChatterjeeNZ17}.
Second, we generate invariants consisting of polynomial \emph{inequalities}, whereas several previous approaches synthesize polynomial \emph{equalities}~\cite{DBLP:conf/popl/SankaranarayananSM04,rodriguez2004automatic}.
Third, our approach is semi-complete, thus it is more accurate than approaches with relaxations (e.g.~\cite{DBLP:conf/vmcai/Cousot05,DBLP:journals/fcsc/LinWYZ14}).
Fourth, compared to previous complete approaches that solve formulas in the first-order theory of reals (e.g.~\cite{chencav15,DBLP:journals/fcsc/YangZZX10,kapur2004automatically}) to generate invariants, our approach has lower complexity, i.e.~our approach is subexponential, whereas they take exponential or doubly-exponential time. Another notable work is~\cite{qlearn} that synthesizes barrier certificates for the verification of reinforcement learning methods. Compared to~\cite{qlearn}, our approach is not restricted to barrier certificates and can handle non-convex invariants, whereas~\cite{qlearn} relies on~\cite{mosek} and requires convexity.

\paragraph{Approaches in Dynamical Systems} Similar techniques have also been applied in the context of continuous and hybrid dynamical systems~\cite{dyn1,dyn2,dyn3}. However, they ensure neither completeness nor subexponential complexity. 

\paragraph{Comparison with~\cite{DBLP:conf/atva/FengZJZX17}} Finally, we compare our approach with the most related work, i.e.~\cite{DBLP:conf/atva/FengZJZX17}. A main difference is that our approach can find a representative set of all solutions, but~\cite{DBLP:conf/atva/FengZJZX17} might miss some solutions, i.e.~it only guarantees to find at least one solution as long as the problem is feasible.
In terms of techniques,~\cite{DBLP:conf/atva/FengZJZX17} uses Stengle's positivstellensatz, while we use Putinar's positivstellensatz and the algorithm of Grigor'ev~\cite{grigoriev1988solving}.
Moreover,~\cite{DBLP:conf/atva/FengZJZX17} considers the class of probabilistic programs without non-determinism and only focuses on single probabilistic while loops, while we consider programs in general form, with non-determinism and recursion, but without probability.

	\section{Illustrative Example}

Before going into technical details, we first illustrate the main ideas and insights behind our approach using a very simple example. Consider the following program:

\lstset{language=prog}
\lstset{tabsize=4}
\begin{lstlisting}[aboveskip=4pt,belowskip=2pt,mathescape,basicstyle=\small,frame=tb]
$\text{Precondition:}$ $100 - y^2 \geq 0$
if $x^2 - 100 \geq 0$ then
	$\text{Invariant:}$ $c_1 \cdot y^2 + c_2 \cdot y + c_3 \geq 0$
	$x$ := $y$
else
	$\text{Invariant:}$ $c_4 \cdot x^2 + c_5 \cdot x + c_6 \geq 0$
	skip
fi
$\text{Postcondition:}$ $c_7 \cdot x + c_8 \geq 0$
\end{lstlisting}

There are two program variables, namely $x$ and $y$. A precondition $100 - y^2 \geq 0$ is assumed to hold at the beginning of the program, and the goal is to synthesize a postcondition and an invariant for each of the branches of the \textbf{if} statement. Moreover, a template is given for each of the desired expressions, e.g.~inside the \textbf{then} branch, we are interested in synthesizing an invariant of the form $c_1 \cdot y^2 + c_2 \cdot y + c_3 \geq 0$, where the $c_i$'s are unknown coefficients, i.e.~the goal is to find values for the $c_i$'s so that this expression becomes an invariant. To do this, it suffices to synthesize values for the $c_i$'s such that the assertion at each point of the program can be deduced from those at its predecessors. More concretely:
\begin{compactenum}[(i)]
	\item $100-y^2 \geq 0 ~\wedge~ x^2-100 \geq 0 \Rightarrow c_1 \cdot y^2 + c_2 \cdot y + c_3 \geq 0$, i.e.~the invariant should hold when we transition inside the \textbf{then} branch.
	\item $100-y^2 \geq 0 ~\wedge~ 100-x^2 > 0 \Rightarrow c_4 \cdot x^2 + c_5 \cdot x + c_6 \geq 0$, i.e.~the invariant should hold when we transition inside the \textbf{else} branch.
	\item $c_1 \cdot y^2 + c_2 \cdot y + c_3 \geq 0 \Rightarrow c_7 \cdot y + c_8 \geq 0$, i.e.~the postcondition should hold when we exit the \textbf{then} branch. Note that the assignment $x := y$ is applied to the RHS.
	\item $c_4 \cdot x^2 + c_5 \cdot x + c_6 \geq 0 \Rightarrow c_7 \cdot x + c_8 \geq 0$, i.e.~the postcondition should hold when we exit the \textbf{else} branch.
\end{compactenum}
One ad-hoc way to satisfy the constraints above is to force the RHS polynomial expression to be a nonnegative combination of the LHS polynomials, e.g.~in (i), we can set $c_1=-1, c_2=0, c_3=100$, essentially making the RHS polynomial equal to the first LHS polynomial. Similarly, in (ii), we can set $c_4=-1, c_5=0, c_6=100.$ However, this cannot work for (iii). To handle this constraint, note that, without loss of generality, we can add any tautology to our assumptions. For example, we know that $(a \cdot y - b)^2 \geq 0$ holds for all real numbers $a$ and $b$, so we prove

 $(a \cdot y - b)^2 \geq 0 ~\wedge~ c_1 \cdot y^2 + c_2 \cdot y + c_3 \geq 0 \Rightarrow c_7 \cdot y + c_8 \geq 0.$
  To solve the latter, we can simply let
  
   $c_7 \cdot y + c_8 = (a \cdot y - b)^2 + d \cdot (c_1 \cdot y^2 + c_2 \cdot y + c_3),$
   
    \noindent where $d$ is a nonnegative real number. Let us expand the RHS to get $c_7 \cdot y + c_8 = a^2 \cdot y^2 - 2 \cdot a \cdot b \cdot y + b^2 + c_1 \cdot d \cdot y^2 + c_2 \cdot d \cdot y + c_3 \cdot d.$ Note that this is an equality between two polynomials over the variable $y$. These polynomials are equal iff \emph{they have the same coefficient for each power of $y$}, therefore this equality is equivalent to the following system:
\begin{compactitem}
	\item $0 = a^2 + c_1 \cdot d$, i.e. the coefficients of $y^2$ should be equal;
	\item $c_7 = - 2 \cdot a \cdot b + c_2 \cdot d$, i.e.~the coefficients of $y$ are equal;
	\item $c_8 = b^2 + c_3 \cdot d$, i.e.~the constant factors should be the same.
\end{compactitem}
We can now use a quadratic programming solver, together with the values we already have for $c_1, c_2, c_3$ from the previous steps, to obtain one possible solution, e.g.~$c_7=-1, c_8=10, a=\frac{1}{2\sqrt{5}}, b=\sqrt{5}, d=\frac{1}{20}.$ We can solve (iv) similarly. Putting everything together, we have:
	\lstset{language=prog}
	\lstset{tabsize=4}
	\begin{lstlisting}[frame=tb,aboveskip=4pt,belowskip=4pt,mathescape,basicstyle=\small]
$\text{Precondition:}$ $100 - y^2 \geq 0$
if $x^2 - 100 \geq 0$ then
	$\text{Invariant:}$ $ -y^2 + 100 \geq 0$
	$x$ := $y$
else
	$\text{Invariant:}$ $ -x^2 + 100 \geq 0$
	skip
fi
$\text{Postcondition:}$ $ 10 - x \geq 0$
	\end{lstlisting}

To obtain this, we had to find values for $c_i$'s and proofs that conditions (i)--(iv) above hold when we plug in these values. The proofs for (i) and (ii) are easy, because the RHS polynomial is already assumed to be nonnegative in the LHS. For (iii) and (iv) we had to become more creative and add suitable tautologies to the LHS. For example, we proved (iii) by showing that $10 - y = \left(\frac{1}{2 \sqrt{5}} \cdot y - \sqrt{5}\right)^2 + \frac{1}{20} \cdot  \left(-y^2+100\right)$, hence the assertion $-y^2+100\geq 0 \Rightarrow 10-y\geq 0$ holds, because the RHS is a nonnegative combination of the LHS and an always-nonnegative polynomial $\left(\frac{1}{2 \sqrt{5}} \cdot y - \sqrt{5}\right)^2.$

Our approach in this paper generalizes the simple ideas above. Given a program, we cannot be sure about the right template to use at each point, so we instead use the most general template, i.e.~our template polynomials contain all possible monomials up to a certain degree. Then, we write the constraints that ensure these templates become a valid inductive invariant (such as (i)--(iv) above). Afterwards, we have to synthesize suitable values for the unknown coefficients ($c_i$'s) and prove that all the required constraints hold. In general, when we want to prove a constraint of the form $g_1\geq 0, g_2 \geq 0, \ldots, g_m \geq 0 \Rightarrow g \geq 0$, where $g$ and $g_i$'s are polynomials, we use a technique similar to what we did for constraint (iii) above and write $g$ as a combination of $g_i$'s and sum-of-square polynomials, i.e.
$
g = h_0 + \sum_{i=1}^m h_i \cdot g_i,
$
where each $h_i$ is a sum of squares and hence always nonnegative. Therefore, wherever $g_i$'s are nonnegative, it trivially follows that $g$ must also be nonnegative. Hence, our approach is sound. Moreover, a classical theorem in real algebraic geometry, called Putinar's Positivstellensatz (Theorem~\ref{thm:putinar}), helps us prove that our approach preserves completeness under certain conditions, i.e.~that any positive $g$ can be written as a combination of $g_i$'s in the form above. Using this idea, we can translate our constraints to quadratic programming in essentially the same manner we handled constraint (iii) above, i.e.~by equating the coefficients of corresponding terms on the two sides of the polynomial equality. In the following sections, we formalize and build on these simple ideas.
	
\section{Polynomial Programs and Invariants}

\subsection{Syntax and Semantics}

We consider non-deterministic recursive programs with polynomial assignments and guards. Our syntax is shown in Figure~\ref{fig:syntax}. The $\star$ denotes non-deterministic branching. See Appendix~\ref{app:syntax} for more details. We fix two disjoint finite sets: the set $\pvars$ of program variables and the set $\fnames$ of functions.
 
 \begin{figure}
 	\begin{centering}
 		\resizebox{9cm}{!}{
 			$\begin{array}{rrl}
 			
 			\obj{prog} & ::= &  \obj{func} ~~\mid~~ \obj{func}~\obj{prog}\\
 			
 			\obj{func} & ::= & \obj{fname}~\kw{$($}~\obj{varlist}~\kw{$)$}~\kw{$\{$}~\obj{stmtlist}~\kw{$\}$} \\ 
 			
 			\obj{varlist} & ::= & \obj{var} ~~\mid~~ \obj{var}~\kw{$,$}~\obj{varlist}\\
 			
 			\obj{stmtlist} & ::= & \obj{stmt}  ~~\mid~~ \obj{stmt}~\kw{$;$}~\obj{stmtlist}\\
 			
 			\obj{stmt} & ::= & \kw{skip} ~~\mid~~ \obj{var}~\kw{$:=$}~\obj{expr}\\
 			&& \mid \kw{if}~\obj{bexpr}~\kw{then}~\obj{stmtlist}~\kw{else}~\obj{stmtlist}~\kw{fi}\\
 			&& \mid \kw{if}~\kw{$\star$}~\kw{then}~\obj{stmtlist}~\kw{else}~\obj{stmtlist}~\kw{fi}\\
 			&& \mid \kw{while}~\obj{bexpr}~\kw{do}~\obj{stmtlist}~\kw{od}\\
 			&& \mid \obj{var}:=\obj{fname}~\kw{$($}~\obj{varlist}~\kw{$)$}\\
 			&& \mid \kw{return}~\obj{expr}	
 			\end{array}$
 		}
 	\end{centering}
 	\caption{Our Syntax. See Appendix~\ref{app:syntax} for more details.}
 	\label{fig:syntax}
 \end{figure}

\paragraph{Program Counters (Labels)} We assign a unique \emph{program counter} to each statement of the program and the endpoint of every function. We also refer to program counters as \emph{labels}. We use $\labels$ to denote the set of labels. We denote the first label in a function $f$ by $\lin^f$ and the label of its endpoint by $\lout^f$.

\paragraph{Types of Labels} We partition the set $\labels$ of labels as follows:
\begin{compactitem}
	\item $\labels_a$: Labels of \textbf{a}ssignment, skip or return statements,
	\item $\labels_b$: Labels of \textbf{b}ranching (if) and while-loop statements,
	\item $\labels_c$: Labels of function \textbf{c}all statements,
	\item $\labels_d$: Labels of non-\textbf{d}eterministic branching statements,
	\item $\labels_e$: Labels of the \textbf{e}ndpoints of functions.
\end{compactitem}

\begin{example} \label{ex:running}
	 Consider the simple program in Figure~\ref{fig:running}. The numbers on the left are the labels and their subscripts denote their types. We will use this program as our running example. It contains a single function $\func{sum}$ that takes a parameter $n$ and then non-deterministically sums up some of the numbers between $1$ and $n$ and returns the summation. Our goal is to prove that the return value of $\func{sum}$ is always less than $0.5\cdot n^2 + 0.5\cdot n + 1$.

\begin{figure}
	\begin{center}
	\begin{minipage}{0.5\linewidth}
	\lstset{language=prog}
	\lstset{tabsize=4}
	\begin{lstlisting}[aboveskip=0pt,belowskip=0pt,mathescape,basicstyle=\small]
	   $\func{sum}$($n$) {
	1$_a$:	$i$ := $1$;
	2$_a$:	$s$ := $0$;
	3$_b$:	while $i \leq n$ do
	4$_d$:		if $\star$ then
	5$_a$:			$s$ := $s+i$
			else
	6$_a$:			skip
			fi;
	7$_a$:		$i$ := $i+1$
		od;
	8$_a$:	return s
	9$_e$: }
	\end{lstlisting}
	\end{minipage}
	\end{center}
\caption{A non-deterministic summation program}
\label{fig:running}
\end{figure}
\end{example}

\paragraph{New Variables} For each function $f \in \fnames$, 
whose header is of the form $f(v_1, \ldots, v_n)$, we define $n+1$ new variables $\ret^f, \bar{v}_1, \ldots, \bar{v}_n$. Informally, $\ret^f$ is the return value of the function $f$ and each variable $\bar{v}_i$ holds the value passed to the function $f$ from its caller for parameter $v_i$ without allowing $f$ to change it. We define $\pvars_*^f := \{ \ret^f, v_1, \ldots, v_n, \bar{v}_1, \ldots, \bar{v}_n \}$ and let $\pvars^f$ contain $\pvars_*^f$, and any other variable that appears in the body of $f$. Similarly, $\labels^f$ denotes the set of labels in $f$.

\paragraph{CFGs} We use standard control flow graphs as in~\cite{cav17,allen1970control}. A \emph{Control Flow Graph (CFG)} is a triple $(\fnames, \labels, \rightarrow)$ where:
\begin{compactitem} 
	
	\item $\fnames$ is the set of functions;
	
	\item
	the labels $\labels$ form the set of vertices, and
	
	\item
	$\rightarrow$ is a relation whose members are triples $(\loc, \alpha, \loc')$ in which the source label $\loc$ and the target label $\loc'$ are in the same $\labels^f$, the source label is not the end of function label, i.e.~$\loc \neq \lout^f$, and $\alpha$ is one of the following: (i)~an update function $\alpha: \vals{f} \rightarrow \vals{f}$ if $\loc \in \labels_a$, or (ii)~a propositional polynomial predicate over $\pvars^f$ if $\loc \in \labels_b$, or (iii)~$\perp$ if $\loc \in \labels_c$, or (iv)~$\star$ if $\loc \in \labels_d$.
\end{compactitem}
Intuitively, we say that a CFG $(\fnames, \labels, \rightarrow)$ is the CFG of program $P$ if (i)~for each label $\loc \in \labels$, the successors of $\loc$ in $\rightarrow$ are the labels that are in the same function as $\loc$ and can possibly be executed right after $\loc$, and (ii)~the $\alpha$'s correspond to the behavior of the program, e.g.~if $(\loc, \alpha, \loc') \in \rightarrow$, $\loc$ is an if statement and $\loc'$ is the first statement in its \kw{else} part, then $\alpha$ should be the negation of the if condition. See~\cite{cav17} for more details. Note that a return statement in a function $f$ changes the value of the variable $\ret^f$ and is succeeded by the endpoint label $\lout^f$. Figure~\ref{fig:cfg} provides the CFG of Example~\ref{ex:running}. Our semantics are defined based on a CFG in the standard manner. See Appendix~\ref{app:semantics} for details.

\begin{figure}
	\begin{center}
		\includegraphics[keepaspectratio,align=t,height=2.5cm]{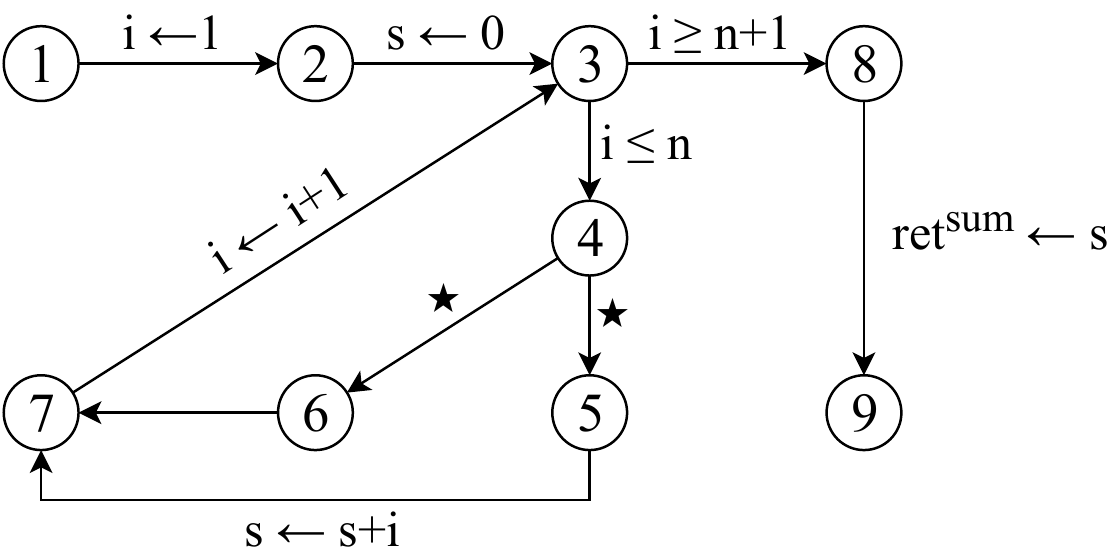}
	\end{center}
	\caption{CFG of the Program in Example~\ref{ex:running} (Figure~\ref{fig:running})}
	\label{fig:cfg}
\end{figure}

\subsection{Invariants} \label{sec:inv}

\paragraph{Pre-conditions} A \emph{pre-condition} is a function $\precond$ mapping each label $\loc \in \locs^f$ of the program to a conjunctive propositional formula $\precond(\loc) := \bigwedge_{i=0}^m \left( \mathfrak{e}_i \geq 0 \right)$, where each $\mathfrak{e}_i$ is an arithmetic expression over the set $\pvars^f$ of variables\footnote{Classically, pre-conditions are only defined for the first labels of functions, but we allow pre-conditions for every label. This setting is strictly more general, given that one can let $\precond(\loc) = \mathbf{true}$ for every other label.}\footnote{The value of every uninitialized (non-parameter) variable $v$ is always $0$ when the program reaches $\lin^f$. Hence, w.l.o.g.~we assume that $\precond(\lin^f)$ contains the assertions $v \geq 0$ and $-v \geq 0$. Similarly, we assume that for every parameter $v$ of $f$, we have the assertions $v-\bar{v}\geq 0$ and $\bar{v}-v\geq 0$ in $\precond(\lin^f)$.}. 
Intuitively, a pre-condition specifies a set of requirements for the runs of the program, i.e.~a run of the program is valid if it always respects the pre-condition, and a run that does not satisfy the pre-condition is considered to be invalid or impossible and is ignored in our analysis.

\paragraph{Post-conditions} A \emph{post-condition} is a function $\postcond$ that maps each program function $f$ of the form $f(v_1, \ldots, v_n)$ to a conjunctive propositional formula $\postcond(f) := \bigwedge_{i=0}^m \left( \mathfrak{e}_i > 0 \right)$ over $\{ \ret^f, \bar{v}_1, \ldots, \bar{v}_n \}$. 
Informally, a post-condition characterizes the return value $\ret^f$ of each function $f$ based on the values of parameters passed to $f$ when it was called.

\begin{remark}[Strict and Non-strict Inequalities]
	In the definitions above the inequalities in post-conditions are strict, whereas pre-conditions contain non-strict inequalities. There is a technical reason behind this choice, having to do with Theorem~\ref{thm:putinar}. Basically, Putinar's positivstellensatz characterizes \emph{strictly positive} polynomials over a \emph{closed} semi-algebraic set. Therefore, this subtle difference in the definitions of pre and post-conditions is necessary for our completeness result (Lemma~\ref{lem:complete}). However, our soundness does not depend on it. 
\end{remark}

\paragraph{Model of Computation} We consider programs in which variables can have arbitrary real values. However, some of our results only hold if the variable values are bounded. In such cases we explicitly mention that the result holds on ``bounded reals''. The formal interpretation of this point is that there exists a constant value $c \in \mathbb{R}^+$ such that for every label $\loc \in \locs^f$ and every variable $v \in \pvars^f,$ the pre-condition $\precond(\loc)$ contains the inequalities $-c \leq v \leq c.$ In other words, in the bounded reals model of computation, a variable overflows if its value becomes more than $c$ (resp.~underflows if its value becomes less than $-c$), and any run containing an overflow or underflow is considered invalid. As a direct consequence, in every valid run, when we are in a function $f$ and the valuation of variables is $\nu$, we have $\norm{\nu}_2 \leq c \sqrt{\vert \pvars^f \vert}.$ Hence, when discussing bounded reals, we assume every pre-condition contains the inequality $\norm{\pvars^f}_2^2 \leq c^2 {\vert \pvars^f \vert}$, too.\footnote{More concretely, if $\pvars^f = \{v_1, v_2, \ldots, v_n\}$, then the pre-condition contains the inequality $v_1^2 + v_2^2 + \ldots + v_n^2 \leq c^2 \cdot n.$} Note that this inequality is entailed by the bounds on values of individual variables. We will later use it to satisfy the requirements of our positivstellensatz (Theorem~\ref{thm:putinar}).

\paragraph{Invariants} An \emph{invariant} is a function $\invariant$ mapping each label $\loc \in \locs^f$ of the program to a conjunctive propositional formula $\inv(\loc) := \bigwedge_{i=0}^m \left( \mathfrak{e}_i > 0 \right)$ over $\pvars^f$, such that whenever a valid run reaches $\loc$, $\inv(\loc)$ is satisfied. 

\paragraph{Positivity Witnesses} Let $\mathfrak{e}$ be an arithmetic expression on program variables and $\phi = \bigwedge_{i=0}^m (\mathfrak{e}_i \bowtie_i 0)$ for $\bowtie_i~\in \{ >, \geq\}$, such that for every valuation $\nu$, we have $\nu \satisfies \phi \Rightarrow \mathfrak{e}(\nu) > 0$. We say that a constant $\epsilon > 0$ is a \emph{positivity witness} for $\mathfrak{e}$ w.r.t.~$\phi$ if for every valuation $\nu$, we have $\nu \satisfies \phi \Rightarrow \mathfrak{e}(\nu) > \epsilon$. In the sequel, we limit our focus to inequalities that have positivity witnesses. Intuitively, this means that we consider invariants of the form $\bigwedge_{j=1}^m (\mathfrak{e}_j > 0)$ where the values of $\mathfrak{e}_j$'s in the runs of the program cannot get arbitrarily close to $0$\footnote{Note that this is a very minor restriction, in the sense that if $\mathfrak{e} > 0$ is an invariant, then so is $\mathfrak{e}+\epsilon>0$. We are unable to find invariants $\mathfrak{e}>0$ where $\mathfrak{e}$ can get arbitrarily close to $0$ over all valid runs of the program. However, in such cases, we can synthesize $\mathfrak{e}+\epsilon>0$ for \emph{any} positive $\epsilon$, as long as $\mathfrak{e}+\epsilon>0$ is also part of an inductive invariant.}. 

\paragraph{Inductive Assertion Maps} An \emph{inductive assertion map} for a non-recursive program is a function $\ind$ mapping each label $\loc \in \locs^f$ of the program to a conjunctive propositional formula $\ind(\loc) := \bigwedge_{i=0}^m \left( \mathfrak{e}_i > 0 \right)$ over $\pvars^f$, such that the following two conditions hold:
\begin{compactitem}
	\item \emph{Initiation.}  $\precond(\lin^f) \Rightarrow \ind(\lin^f)$
	\item \emph{Consecution.} In every valid transition of the program from $\loc_0$ to $\loc_1$, if $\ind(\loc_0)$ holds at $\loc_0$, then $\ind(\loc_1)$ must hold at $\loc_1$. 
	Intuitively, this condition means that the inductive assertion map cannot be falsified by running a valid step of the execution of the program.
\end{compactitem}

It is well-known that every inductive assertion map is an invariant. So, inductive assertion maps are often called \emph{inductive invariants}, too. See Appendix~\ref{app:invariant} for a short proof and more formal definitions.

\begin{example} \label{ex:prepostinv}
	Consider the program in Figure~\ref{fig:running}. Assuming that we have the pre-condition $n \geq 0$ at label $1$, it is easy to show that for any $\epsilon>0$, $\ind(\loc) := (n+\epsilon > 0 \wedge i+\epsilon > 0 \wedge s+\epsilon > 0)$ for all $\loc \in \{1, \ldots, 9\}$ is an inductive assertion map, i.e.~it holds at the beginning of the program and no valid execution step falsifies it. Hence, it is also an invariant. Moreover, using this invariant, one can prove the post-condition $\ret^\func{sum} + \epsilon > 0.$
\end{example}

We extend the notion of inductive assertion maps to recursive programs using abstract paths:

\paragraph{Abstract Paths} Informally, an \emph{abstract path} is an execution path in which all function calls are abstracted, i.e.~removed and replaced with a simple transition in the parent function that respects the pre and post-condition of the called function. The intuition behind an abstract path is to use the pre/post-condition of a function as an over-approximation of its behavior, and hence avoid running the function itself. See Appendix~\ref{app:semantics} for a formal definition. 

\begin{example}
	Assume the pre and post-conditions of Example~\ref{ex:prepostinv} and consider a program $P$ whose main function $\mainFunction$ calls the function $\func{sum}$ of Figure~\ref{fig:running}. Suppose  $\func{sum}$ is called at label $\loc$, i.e.~the statement at $\loc$ is $y := \func{sum}(x)$, and $(\loc, \perp, \loc') \in \rightarrow.$ In a normal run of $P$, when the program reaches $\loc$, 
	control moves to $\func{sum}.$ In contrast, in an abstract path starting at $\loc$,
	 control directly moves to $\loc'$, provided that no variable other than $y$ gets its value changed and that the pre-condition and post-condition are satisfied. For example, the following sequences are abstract paths:
	 
	$
	\begin{matrix}
	\langle \langle \left( \mainFunction, \loc, x=3, y=0 \right) \rangle, \langle \left( \mainFunction, \loc', x=3, y=\epsilon \right) \rangle \rangle\\
	\langle \langle \left( \mainFunction, \loc, x=3, y=1 \right) \rangle, \langle \left( \mainFunction, \loc', x=3, y=99.9 \right) \rangle \rangle
	\end{matrix}
	$

\smallskip
The latter configuration cannot happen in any valid run, but it does not violate the conditions of an abstract path. This is because the post-condition in this example is very weak and hence abstract paths grossly overestimate valid paths. As we will see, our algorithms synthesize stronger post-conditions as part of the invariant generation process. Finally, the following are \emph{not} abstract paths:

	$
	\begin{matrix}
	\langle \langle \left( \mainFunction, \loc, x=-1, y=1 \right) \rangle, \langle \left( \mainFunction, \loc', x=-1, y=10 \right) \rangle \rangle \\ \footnotesize{\text{Reason: It violates }\precond(\lin^\func{sum})[n \leftarrow x, \bar{n} \leftarrow x]} \vspace{2mm}\\ 

	\langle \langle \left( \mainFunction, \loc, x=1, y=1 \right) \rangle, \langle \left( \mainFunction, \loc', x=1, y=-1 \right) \rangle \rangle \\ \footnotesize{\text{Reason: It violates }\postcond(\func{sum})[\bar{n} \leftarrow x, \ret^\func{sum} \leftarrow y]} \vspace{2mm}\\
	
	\langle \langle \left( \mainFunction, \loc, x=3, y=3 \right) \rangle, \langle \left( \mainFunction, \loc', x=2, y=4 \right) \rangle \rangle \\ \footnotesize{\text{Reason: It changes the value of } x}
	\end{matrix}
	$
\end{example}

\paragraph{Recursive Inductive Invariants} Given a recursive program $P$ and a pre-condition $\precond$, a \emph{recursive inductive invariant} is a pair $(\postcond, \ind)$ where $\postcond$ is a post-condition and $\ind$ is a function that maps every label $\loc \in \locs^f$ of the program to a conjunctive propositional formula $\ind(\loc) := \bigwedge_{i=0}^m \left( \mathfrak{e}_i > 0 \right)$, such that the following requirements are met:
\begin{compactitem}
	\item \emph{Initiation.} For every function $f$, we have $ \precond(\lin^f) \Rightarrow \ind(\lin^f).$
	\item \emph{Consecution.} For every valid unit-length \emph{abstract} path that transitions from $\loc_0$ to $\loc_1$, if $\ind(\loc_0)$ holds at $\loc_0$, then $\ind(\loc_1)$ must hold at $\loc_1$.
	
	\item \emph{Post-condition Consecution.} For every valid unit-length \emph{abstract} path that starts at $\loc_0 \in \locs^f$ and ends at the endpoint label $\lout^f$, if $\ind(\loc_0)$ holds at $\loc_0$, then the post-condition $\postcond(f)$ must hold at $\lout^f.$
\end{compactitem}
Following an argument similar to the case of inductive invariants, if $(\postcond, \ind)$ is a recursive inductive invariant, then $\ind$ is an invariant. See Appendix~\ref{app:invariant} for details.

We define our synthesis problem in terms of (recursive) inductive invariants, because the classical method for finding or verifying invariants is to consider inductive invariants that strengthen them~\cite{sriram,DBLP:books/daglib/0080029}.

\paragraph{The Invariant Synthesis Problem} \label{def:prob}
Given a program $P$, together with a pre-condition $\precond$, the invariant synthesis problem asks for (recursive) inductive invariants of a given form and size (e.g.~linear or polynomial of a given degree). The problem can be divided into two variants:
\begin{compactitem}
	\item The \emph{Strong Invariant Synthesis Problem} asks for a characterization or a representative set of all possible invariants.
	\item The \emph{Weak Invariant Synthesis Problem} provides an objective function over the invariants (e.g.~a function over the coefficients of polynomial invariants) and asks for an invariant that maximizes the objective function.
\end{compactitem} 

\paragraph{Motivation} There are several motives for defining both strong and weak invariant synthesis:
\begin{compactitem}
	\item Strong invariant synthesis can be used in compositional reasoning, e.g.~having separate representations for \emph{all} invariants of programs $P_1, P_2, \ldots, P_n,$ an invariant for their sequential composition $P_1; P_2; \ldots; P_n$ can be derived by considering each part separately.
	\item Strong invariant synthesis is computationally expensive, so in practice, weak invariant synthesis can be used to obtain invariants that are desirable (according to a given objective function). For example, in Section~\ref{sec:exp}, we use it to prove desired assertions (partial invariants) at a few points of the program by synthesizing an inductive invariant that includes them.
	\item Another use-case of weak invariant synthesis is to find bounds for a given expression $R$ at some point of the program. The objective function can be set to find the tightest possible bound. Such bounds are useful in many contexts, e.g.~if $R$ is a ranking function, then its upperbound is also a bound on the runtime of the program.
\end{compactitem} 

\paragraph{Polynomial Invariants} In the sequel, we consider the synthesis problems for \emph{polynomial} invariants and pre and post-conditions, i.e.~we assume that all arithmetic expressions used in the atomic assertions are polynomials.

	\section{Invariants for Non-recursive Programs} \label{sec:algo-nonrecursive}

 We first provide a sound and semi-complete reduction from inductive invariants to solutions of a system of quadratic equalities. Our main tool is a theorem in real semi-algebraic geometry called Putinar's positivstellensatz~\cite{putinar1993positive}.
We show that the Strong Invariant Synthesis problem can be solved in subexponential time. We also show that the Weak Invariant Synthesis problem can be reduced to QCLP.

\subsection{Mathematical Tools and Lemmas} \label{app:SOS} \label{sec:math}

 The following theorem is the main tool in our reduction:

\begin{theorem}[Putinar's Positivstellensatz~\cite{putinar1993positive}] \label{thm:putinar} Let $V$ be a finite set of variables and $g, g_1, \ldots, g_m \in \mathbb{R}[V]$ polynomials over $V$ with real coefficients. We define $\Pi := \{ x \in \mathbb{R}^V ~\mid~\forall i~~ g_i(x) \geq 0  \}$ as the set of points in which every $g_i$ is non-negative. If (i)~there exists some $g_k$ s.t. the set $\{ x \in \mathbb{R}^V ~\mid~ g_k(x) \geq 0  \}$ is compact, and (ii)~$g(x)>0$ for all $x \in \Pi$, then
	\begin{equation} \label{eq:putinar}
		g = h_0 + \sum_{i=1}^m h_i \cdot g_i
	\end{equation}
where each polynomial $h_i$ is the sum of squares of some polynomials in $\mathbb{R}[V]$, i.e.~$h_i = \sum_{j=0}^{n} f_{i,j}^2$ for some $f_{i,j}$'s in $\mathbb{R}[V]$.
\end{theorem}

\begin{corollary} [Proof in Appendix~\ref{app:col:put}]\label{col:put}
	Let $V, g, g_1, \ldots, g_m$ and $\Pi$ be as above. Then $g(x) > 0$ for all $x \in \Pi$ \emph{if and only if}:
	\begin{equation} \label{eq:colput}
	g = \epsilon + h_0 + \sum_{i=1}^m h_i \cdot g_i
	\end{equation}
	where $\epsilon >0$ is a real number and each polynomial $h_i$ is the sum of squares of some polynomials in $\mathbb{R}[V]$.
\end{corollary}

Hence, Putinar's positivstellensatz provides a characterization of all polynomials $g$ that are positive over the closed set $\Pi$. Intuitively, given a set of atomic non-negativity assumptions $g_i(x) \geq 0$, in order to find all  polynomials $g$ that are positive under these assumptions, we only need to look into polynomials of form~(\ref{eq:colput}). Moreover, the real number $\epsilon$ in $(\ref{eq:colput})$ serves as a positivity witness for $g$.

Our algorithm also relies on the following lemma:
\begin{lemma}[Proof in Appendix~\ref{app:maths2}] \label{lm:sq}
	Given a polynomial $h \in \mathbb{R}[V]$ as input, the problem of deciding whether $h$ is a sum of squares, i.e.~whether $h$ can be written as $\sum_{i} f_i^2$ for some polynomials $f_i \in \mathbb{R}[V],$ can be reduced in polynomial time to solving a system of quadratic equalities. 
\end{lemma}

\subsection{Overview of the Approach}

In this section, we provide an overview of our algorithms. The next sections go through all the details. Our algorithms for Strong and Weak Invariant Synthesis are very similar. They each consist of four main steps and differ only in the last step. The steps are as follows:
\begin{compactenum}
	\item[Step 1)] First, the algorithm creates a template for the inductive invariant at each label. More specifically, it creates polynomial templates of the desired size and degree, but with \emph{unknown coefficients}. The goal is to synthesize values for these unknown coefficients so that the template becomes a valid inductive invariant.
	\item[Step 2)] The algorithm generates a set of constraints that should be satisfied by the template so as to ensure that it becomes an inductive invariant. These constraints encode the initiation and consecution requirements as in the definition of inductive invariants. Moreover, they have a very specific form: each constraint consists of polynomials $g_1, \ldots, g_m$ and $g$ and encodes the requirement that for every valuation $\nu$, if we have $g_1(\nu) \geq 0, g_2(\nu) \geq 0, \ldots, g_m(\nu) \geq 0$, then we must also have $g(\nu) > 0$.
	\item[Step 3)] Exploiting the structure of the constraints generated in the previous step, the algorithm applies Putinar's positivstellensatz to translate the constraints into quadratic equalities over the unknown coefficients.
	\item[Step 4)] The algorithm uses an external solver for handling the system of quadratic equalities generated in the previous step. In case of Strong Invariant Synthesis, the external solver would use the algorithm of~\cite{grigoriev1988solving} to provide a representative set of all invariants. In contrast, for Weak Invariant Synthesis, the external solver is an optimization suite for quadratic programming (QCLP).
\end{compactenum}

\subsection{Strong Invariant Synthesis}

We now provide a formal description of the input to our algorithm for Strong Invariant Synthesis and then present details of every step. 

\paragraph{The $\strongInvAlgo$ Algorithm} We present an algorithm $\strongInvAlgo$ that gets the following items as its input:
\begin{compactitem}
	\item A non-recursive program $P$ which is generated by the grammar in Figure~\ref{fig:syntax},
	\item A polynomial pre-condition $\precond$,
	\item Positive integers $d,$ $n$ and $\Upsilon$, where $d$ is the degree of polynomials in the desired inductive invariants, $n$ is the desired size of the invariant generated at each label, i.e.~number of atomic assertions, and $\Upsilon$ is a technical parameter to ensure semi-completeness, which will be discussed later;
\end{compactitem} 
and produces a representative set of all inductive invariants $\ind$ of the program $P$, such that for all $\loc \in \locs$, the set $\ind(\loc)$ consists of $n$ atomic assertions of degree at most $d$. Our algorithm consists of the following four steps:

\smallskip\noindent\textbf{Step 1) Setting up templates.} Let $\pvars^f = \{ v_1, v_2, \ldots, v_t\}$ and define $\monomials_d^f = \{m_1, m_2, \ldots, m_r\}$ as the set of all monomials of degree at most $d$ over $\pvars^f$, i.e.~$\monomials_d^f :=  \{ \prod_{i=1}^t v_i^{\alpha_i} ~\mid~ \forall i~~\alpha_i\in \mathbb{N}_0 ~\wedge~ \sum_{i=1}^t \alpha_i \leq d \}$. At each label $
\loc \in \locs^f$ of the program $P$, the algorithm generates a template $\eta(\loc) := \bigwedge_{i=1}^n \varphi_{\loc, i}$ where each $\varphi_{\loc, i}$ is of the form $\varphi_{\loc, i} := \left( \sum_{j=1}^r s_{\loc, i, j} \cdot m_j > 0 \right)$. Here, the $s_{\loc, i, j}$'s are new unknown variables. For brevity, we call them $s$-variables. Intuitively, our goal is to synthesize values for $s$-variables such that $\eta$ becomes an inductive invariant.

\begin{example} Consider the summation program in Figure~\ref{fig:running}. We have $\pvars^\func{sum} = \{n, \bar{n}, i, s, \ret^\func{sum} \}$. For brevity we define $r := \ret^\func{sum}$. Suppose that we want to synthesize a single quadratic assertion as the invariant at each label. In Step~1, the algorithm creates the following template for each label $\loc \in \{1, 2, \ldots, 9\}$: 
	
	\begin{equation*} 
	\scalebox{0.85}{$
	\begin{split}
	\eta(\loc) & :=~ s_{\loc, 1, 1} + s_{\loc, 1, 2} \cdot n + s_{\loc, 1, 3} \cdot \bar{n} + s_{\loc, 1, 4} \cdot i + s_{\loc, 1, 5} \cdot s  + s_{\loc, 1, 6} \cdot r +  \\
	& s_{\loc, 1, 7} \cdot n^2 + s_{\loc, 1, 8} \cdot n \cdot \bar{n} + s_{\loc, 1, 9} \cdot n \cdot i + s_{\loc, 1, 10} \cdot n \cdot s + s_{\loc, 1, 11} \cdot n \cdot r +\\
	& s_{\loc, 1, 12} \cdot \bar{n}^2 + s_{\loc, 1, 13} \cdot \bar{n} \cdot i + s_{\loc, 1, 14} \cdot \bar{n} \cdot s + s_{\loc, 1, 15} \cdot \bar{n} \cdot r+ s_{\loc, 1, 16} \cdot i^2 +\\
	& s_{\loc, 1, 17} \cdot i \cdot s + s_{\loc, 1, 18} \cdot i \cdot r+ s_{\loc, 1, 19} \cdot s^2 + s_{\loc, 1, 20} \cdot s \cdot r + s_{\loc, 1, 21} \cdot r^2 > 0.\\
 	\end{split}$}
	\end{equation*}
\end{example}

\smallskip\noindent\textbf{Step 2) Setting up constraint pairs.} For each transition $e = (\loc, \alpha, \loc')$ of the CFG of $P$, the algorithm constructs a set $\Lambda_e$ of \emph{constraint pairs} of the form $\lambda = (\Gamma, g)$ where $\Gamma = \bigwedge_{i=1}^m \left(g_i \geq 0\right)$ and $g, g_1, \ldots, g_m$ are polynomials with unknown coefficients (based on the $s$-variables).   Intuitively, a condition pair $(\Gamma, g)$ encodes the following condition: \begin{equation*}
\scalebox{0.85}{$
\forall \nu \in \mathbb{R}^{f} ~~\nu \satisfies \Gamma \Rightarrow g(\nu) > 0~~ \equiv ~~\forall \nu \in \mathbb{R}^{\pvars^f}  \left( \forall g_i \in \Gamma ~~ g_i(\nu)\geq 0 \right) \Rightarrow g(\nu) > 0.$}
\end{equation*}
The construction is as follows (note that all computations are done symbolically):
\begin{compactitem}
	\item If $\loc \in \locs_a$, for every polynomial $g$ for which $g>0$ appears in $\eta(\loc')$, the algorithm adds the condition pair $(\precond(\loc) \wedge \eta(\loc) \wedge (\precond(\loc') \circ \alpha), g \circ \alpha)$ to $\Lambda_e$. Note that $\alpha$ is an update function that assigns a polynomial to every variable and hence the constraint pair can be computed symbolically.
	\item If $\loc \in \locs_b$, then $\alpha$ is a propositional predicate. The algorithm  writes $\alpha$ in disjunctive normal form as $\alpha = \alpha_1 \vee \alpha_2 \vee \ldots \vee \alpha_a$. Each $\alpha_i$ is a conjunction of atomic assertions. For every $\alpha_i$ and every $g$ such that $g>0$ appears in $\eta(\loc')$, it adds the condition pair $(\precond(\loc) \wedge \eta(\loc) \wedge \precond(\loc') \wedge \alpha_i, g)$ to $\Lambda_e$. 
	\item If $\loc \in \locs_d$, for every $g$ for which $g>0$ appears in $\eta(\loc')$, it adds the condition pair $(\precond(\loc) \wedge \eta(\loc) \wedge \precond(\loc'), g)$ to $\Lambda_e$.
	
\end{compactitem}
 Finally, the algorithm constructs the following set $\Lambda_\text{in}$:
\begin{compactitem}
	\item For every polynomial $g$ for which $g>0$ appears in $\eta(\lin^f)$, the algorithm constructs the constraint pair $(\precond(\lin^f), g)$ and adds it to $\Lambda_\text{in}$.
\end{compactitem}

\begin{example} \label{ex:step2}
	In the summation program of Figure~\ref{fig:running}, suppose that $\precond(1) := (n \geq 0) \wedge (i \geq 0) \wedge (-i\geq 0) \wedge (s \geq 0) \wedge (-s \geq 0) \wedge (\ret^\func{sum} \geq 0) \wedge (-\ret^\func{sum} \geq 0) \wedge (n - \bar{n} \geq 0) \wedge (\bar{n}-n\geq 0)$ and $\precond(\loc) := (1 \geq 0) \equiv \mathbf{true}$ for every $\loc \neq 1$. Note that $(n \geq 0)$ is the only non-trivial assertion and all the other assertions are true by definition, given that $1$ is the first statement in $\func{sum}$. We provide some examples of constraint pairs generated in Step 2 of the algorithm:
	\begin{compactitem}
		\item $1 \in \locs_a$ and $e_1 = (1, [i \leftarrow 1], 2) \in \rightarrow$ (see the CFG in Figure~\ref{fig:cfg}). Hence, we have the following constraint pair:
		$$
			(\precond(1) \wedge \eta(1) \wedge \precond(2)[i \leftarrow 1], \eta(2)[i \leftarrow 1])
		$$
		which is symbolically computed as:
		\begin{equation*}
		\left( \begin{matrix}
		\makecell{
		(n \geq 0) \wedge (i \geq 0) \wedge (-i \geq 0) \wedge \\
		(s \geq 0) \wedge (-s \geq 0) \wedge (r \geq 0) \wedge \\
		(-r \geq 0) \wedge (n-\bar{n} \geq 0) \wedge (\bar{n}-n \geq 0) \wedge\\
		(s_{1,1,1} + s_{1, 1, 2} \cdot n + \ldots + s_{1,1,21} \cdot r^2 \geq 0) } & , \end{matrix} \right.
	\end{equation*}
	
	\begin{equation*}
	\left. \begin{matrix}
	  \makecell{s_{2,1,1} + s_{2,1,2} \cdot n + \ldots + s_{2,1,4} +\ldots +\\ s_{2,1,9} \cdot n + \ldots + s_{2,1,13} \cdot \bar{n} + s_{2, 1, 16} +\\ s_{2, 1, 17} \cdot s + s_{2, 1, 18} \cdot r + \ldots + s_{2, 1, 21} \cdot r^2}
		\end{matrix} \right)
		\end{equation*}
		and added to $\Lambda_{e_1}$. Note that $\precond(2)$ is $\mathbf{true}$ and so ignored.
		
		\item $3 \in \locs_b$ and $e_2 = (3, (n-i \geq 0) ,4) \in \rightarrow$, so the constraint pair $(\precond(3) \wedge \eta(3) \wedge \precond(4) \wedge (n-i \geq 0), \eta(4)) \equiv (\eta(3) \wedge (n-i \geq 0), \eta(4))$ is symbolically computed and added to $\Lambda_{e_2}.$
		
		\item $1 = \lin^\func{sum}$, so the constraint pair $(\precond(1), \eta(1))$ is symbolically computed and added to $\Lambda_{\text{in}}.$
		
	\end{compactitem}
	
\end{example}

\smallskip\noindent\textbf{Step 3) Translating constraint pairs to quadratic equalities.} Let $\Lambda := \bigcup_{e \in \rightarrow} \Lambda_e \cup \Lambda_\text{in}$ be the set of all constraint pairs from the previous step. For each $\lambda = \left( \bigwedge_{i=1}^m \left(g_i \geq 0\right), g \right) \in \Lambda$, the algorithm takes the following actions:
\begin{compactenum}[(i)]
	\item Let $V = \{ v_1, \ldots, v_{t'} \}$ be the set of all program variables that appear in $g$ or the $g_i$'s. The algorithm computes the set $\monomials_\Upsilon = \{ m'_1, m'_2, \ldots, m'_{r'} \}$ of all monomials of degree at most $\Upsilon$ over $V.$ Note that $\Upsilon$ is a technical parameter that was supplied as part of the input.
	\item  It symbolically computes an equation of the form~(\ref{eq:colput}):
	$$\begin{matrix}g = \epsilon + h_0 + \sum_{i=1}^m h_i \cdot g_i& & (\dagger)\end{matrix}$$
	where $\epsilon$ is a new unknown and positive real variable and each polynomial $h_i$ is of the form $\sum_{j=1}^r t_{i,j} \cdot m'_j$. Here, the $t_{i,j}$'s are also new unknown variables. Intuitively, we aim to synthesize values for both $t$-variables and $s$-variables in order to ensure the polynomial equality~$(\dagger)$. Note that both sides of $(\dagger)$ are polynomials in $\mathbb{R}[V]$ whose coefficients are quadratic expressions over the newly-introduced $s$-, $t$- and $\epsilon$-variables. 
	\item The algorithm equates the coefficients of corresponding monomials in the left and right hand sides of $(\dagger)$, leading to a set of quadratic equalities over the new variables.
	\item The algorithm computes a set of quadratic equalities which are equivalent to the assertion that the $h_i$'s can be written as sums of squares (Lemma~\ref{lm:sq}). 
\end{compactenum}
The algorithm conjunctively compiles all the generated quadratic equalities into a single system. Note that this system's size is polynomially dependent on the number of lines in the program, assuming that $d$ and $\Upsilon$ are constants.

\begin{remark} \label{rem:ups}
	Based on above, the technical parameter $\Upsilon$ is the maximum degree of the sum-of-squares polynomials $h_i$ in $(\dagger).$ More specifically, in Step~3, we are applying a special case of Putinar's positivstellensatz, in which the sum-of-square polynomials can have a degree of at most $\Upsilon.$
\end{remark}

\begin{example} \label{ex:step3}
	Consider the first constraint pair generated in Example~\ref{ex:step2}. The algorithm writes $(\dagger)$, i.e.~$g = \epsilon + h_0 + \sum_{i=1}^{10} h_i \cdot g_i$ where $g = s_{2,1,1} + \ldots + s_{2, 1, 21} \cdot r^2$ (the polynomial in the second component of the constraint pair), $g_1 = n$, $g_2 = i$, $g_3 = -i$, \ldots, $g_{10} = s_{1,1,1}+ \ldots + s_{1,1,21}$ (the polynomials in the first component of the constraint pair) and each $h_i$ is a newly generated polynomial containing all possible monomials of degree at most $\Upsilon$, e.g.~if $\Upsilon=2$, we have $h_i = t_{i,1} + t_{i, 2} \cdot n + \ldots + t_{i, 21} \cdot r^2$, where each $t_{i,j}$ is a new unknown variable. It then equates the coefficients of corresponding monomials on the two sides of $(\dagger)$. For example, consider the monomial $r^2$. Its coefficient in the LHS of $(\dagger)$ is $s_{2,1,21}$. In the RHS of $(\dagger)$, there are a variety of ways to obtain $r^2$, hence its coefficient is the sum of the following:
	\begin{compactitem}
		\item $t_{0, 21}$, i.e.~the coefficient of $r^2$ in $h_0$,
		\item $t_{6, 6}$, i.e.~the coefficient of $r^2$ in $h_6 \cdot g_6 = h_6 \cdot r$,
		\item $-t_{7, 6}$, i.e.~the coefficient of $r^2$ in $h_7 \cdot g_7 = h_7 \cdot (-r)$,
		\item $t_{10,1} \cdot s_{1,1,21} + t_{10, 6} \cdot s_{1,1,6} + t_{10, 21} \cdot s_{1,1,1}$, i.e.~the coefficient of $r^2$ in $h_{10} \cdot g_{10}$.
	\end{compactitem}
	Hence, the algorithm generates the quadratic equality $s_{2,1,21} = t_{0,21} + t_{6,6} - t_{7, 6} + t_{10,1} \cdot s_{1,1,21} + t_{10, 6} \cdot s_{1,1,6} + t_{10, 21} \cdot s_{1,1,1}$ over the $s-$ and $t-$variables. The algorithm computes similar equalities for every other monomial.
\end{example}

\smallskip\noindent\textbf{Step 4) Finding representative solutions.} The previous step has created a system of quadratic equalities over $s$-variables and other new variables. In this step, the algorithm finds a representative set $\Sigma$ of solutions to this system by calling an external solver. Then, for each solution $\sigma \in \Sigma$, it plugs the values synthesized for the $s$-variables into the template $\eta$ to obtain an inductive invariant $\eta_\sigma := \eta[s_{\loc,i,j} \leftarrow \sigma(s_{\loc,i,j})]$. The algorithm outputs $I = \{\eta_\sigma~\mid~\sigma \in \Sigma \}$.

\begin{remark}[Representative Solutions]
In real algebraic geometry, a standard notion for a representative set of solutions to a polynomial system of equalities is to include one solution from each connected component of the set of solutions~\cite{basu2007algorithms}. The classical algorithm for this problem is called \emph{cylindrical algebraic decomposition} and has a doubly-exponential runtime~\cite{basu2007algorithms,sturmfels2002solving}. However, if the coefficients are limited to rational numbers instead of real numbers, then a subexponential algorithm is provided in~\cite{grigoriev1988solving}\footnote{No tight runtime analysis is available for this algorithm, but~\cite{grigoriev1988solving} proves that its runtime is subexponential.}. Hence, Step 4 of $\strongInvAlgo$ has subexponential runtime in theory.
\end{remark}

\begin{lemma}[Soundness]\label{lem:sound}
	Every output of $\strongInvAlgo$ is an inductive invariant. More generally, for every solution $\sigma \in \Sigma$ obtained in Step 4, the function $\eta_\sigma := \eta[s_{\loc,i,j} \leftarrow \sigma(s_{\loc,i,j})]$ is an inductive invariant.
\end{lemma}

\begin{proof}
	The valuation $\sigma$ satisfies the system of quadratic equalities obtained in Step 3. Hence, for every constraint pair $(\Gamma, g) \in \Lambda$, $g[s_{\loc,i,j} \leftarrow \sigma(s_{\loc,i,j})]$ can be written in the form $(\dagger)$. Hence, we have $\sigma \satisfies (\Gamma, g)$. By definition of Step 2, this is equivalent to $\eta_\sigma$ having the initiation and consecution properties and hence being an inductive invariant. \qed
\end{proof}

We now prove our completeness result. Our approach is semi-complete for bounded reals in the sense of~\cite{cav2016}. Concretely, this means that if we assume the bounded reals model of computation (see Section~\ref{sec:inv}), then any valid inductive invariant can be found by our approach so long as the technical parameter $\Upsilon$ is large enough. Recall that $\Upsilon$ is a bound on the degree of the sum-of-square polynomials (see Remark~\ref{rem:ups}).

\begin{lemma}[Semi-completeness with Compactness] \label{lem:complete} If the pre-condition $\precond$ satisfies the compactness condition of Theorem~\ref{thm:putinar}, i.e.~if in every label $\loc$, $\precond(\loc)$ contains an atomic proposition of the form $g\geq 0$ such that the set $\{ \nu \in \mathbb{R}^{\pvars^f}~\mid~ g(\nu) \geq 0 \}$ is compact, then for every inductive invariant $\ind$ that has the form of the template $\eta$, there exists a natural number $\Upsilon_\ind$, such that for every technical parameter $\Upsilon \geq \Upsilon_\ind$, the invariant $\ind$		 corresponds to a solution of the system of quadratic equalities obtained in Step~3 of $\strongInvAlgo$.
\end{lemma}
\begin{proof}
	Let $\ind$ be an inductive invariant in the form of the template $\eta$. We denote the value of $s_{\loc,i,j}$ in $\ind$ by $\sigma(s_{\loc,i,j})$. Given that $\ind$ satisfies initiation and consecution, the valuation $\sigma$ satisfies every constraint pair $(\Gamma, g)$ generated in Step 2.
	Each such $\Gamma$ contains an assertion $g_i \geq 0$ s.t.~$\{ x \in \vals{\pvars^f} ~\mid~ g_i(x) \geq 0 \}$ is compact. Hence, by Corollary~\ref{col:put}, $g$ can be written in the form~$(
	\dagger)$\footnote{Theorem~\ref{thm:putinar} requires compactness and so does Corollary~\ref{col:put}.} and for large enough $\Upsilon$, there exists a solution to the system  that maps each $s_{\loc,i, j}$ to $\sigma(s_{\loc,i,j}).$ \qed
\end{proof}

\begin{remark}[Bounded Reals, Compactness and Real-world programs] \label{rem:compact}
	Note that in the bounded reals model of computation, every pre-condition enforces that the value of every variable is between $-c$ and $c$ and also contains the polynomial inequality $\norm{\pvars^f}_2^2 \leq c^2  \cdot {\vert \pvars^f \vert}$ (see Section~\ref{sec:inv}). The set of valuations that satisfy the latter polynomial are points in $\vals{f}$ whose distance from the origin is at most a fixed amount $c \sqrt{\vert \pvars^f \vert}.$ Hence, this set is closed and bounded and therefore compact, and satisfies the requirement of Putinar's positivstellensatz. So, our approach is semi-complete for bounded reals.
	It is worth mentioning that almost all real-world programs have bounded variables, e.g.~programs that use floating-point variables can at most store a finite number of  values in each variable, hence their variables are always bounded. Also, note that while the completeness result is dependent on bounded variables, our soundness result holds for general unbounded real variables.
\end{remark}

\begin{remark}[Non-strict inequalities] \label{rem:no-str}
	Although we considered invariants consisting of inequalities with positivity witnesses, i.e.~invariants of the form $\bigwedge (g(x) > 0)$, our algorithm can easily be extended to generate invariants with non-strict inequalities, i.e. invariants of the form $\bigwedge (g(x) \geq 0)$. To do so, it suffices to replace Equation~$(\dagger)$ in Step~3 of the algorithm with Equation~(\ref{eq:putinar}), i.e.~remove the $\epsilon$-variables (positivity witnesses). This results in a sound, but not complete, method for generating non-strict polynomial invariants. Alternatively, we can use Stengle's positivstellensatz~\cite{stengle1974nullstellensatz} instead of Theorem~\ref{thm:putinar}. Stengle is able to characterize non-negative polynomials as well. Hence, using it will ensure semi-completeness even for non-strict invariants. The downside is that, in comparison with Putinar, it leads to a much higher runtime in practice.
\end{remark}

\begin{remark}[Complexity]
	It is straightforward to verify that Steps 1--3 of $\strongInvAlgo$ have polynomial runtime. Hence, our algorithm provides a polynomial reduction from the Strong Invariant Synthesis problem to the problem of finding representative solutions of a system of quadratic equalities. As mentioned earlier, this problem is solvable in subexponential time~\cite{grigoriev1988solving}. Hence, the runtime of our approach is subexponential, too. Note that we consider $d$ and $\Upsilon$ to be fixed constants.
\end{remark}

\begin{theorem}[Strong Invariant Synthesis] \label{thm:sis}
	Given a non-recursive program $P$ and a pre-condition $\precond$ that satisfies the compactness condition, the $\strongInvAlgo$ algorithm solves the Strong Invariant Synthesis problem in subexponential time. This solution is sound and semi-complete.
\end{theorem}

\begin{remark}[Inefficiency] \label{rem:inef}
	Despite its subexponential runtime, the algorithm of~\cite{grigoriev1988solving} has a poor performance in practice~\cite{comparison}. Hence, Theorem~\ref{thm:sis} can only be considered as a theoretical contribution and is not applicable to real-world programs.
\end{remark}

\subsection{Weak Invariant Synthesis and Practical Method}

Due to the practical inefficiency mentioned in Remark~\ref{rem:inef}, in this section we focus on using a very similar approach to reduce the Weak Invariant Synthesis problem to QCLP. Given that there are many industrial solvers capable of handling real-world instances of QCLP, this reduction will provide a practical sound and semi-complete method for polynomial invariant generation. 
We now provide an algorithm for the Weak Invariant Synthesis problem. This is very similar to $\strongInvAlgo$, so we only describe the differences. 

\paragraph{The $\weakInvAlgo$ Algorithm}  Our algorithm $\weakInvAlgo$ takes the same set of inputs as $\strongInvAlgo$, as well as an objective function $\textsf{obj}$ over the resulting inductive invariants. We assume that $\textsf{obj}$ is a linear or quadratic polynomial over the $s$-variables in the template. Intuitively, $\textsf{obj}$ serves as a measure of desirability of a synthesized invariant and the goal is to find the most desirable invariant. 

The first three steps of the algorithm are the same as $\strongInvAlgo$. The only difference is in Step 4, where $\weakInvAlgo$ needs to find only one solution for the computed system of quadratic equalities, i.e.~the solution that maximizes $\textsf{obj}$. Hence, Step 4 is changed as follows:

\smallskip\noindent\textbf{Step 4) Finding the optimal solution.} Step 3 has generated a system of quadratic equalities. In this step, the algorithm uses a QCLP-solver to find a solution $\sigma$ of this system that maximizes the objective function $\textsf{obj}$. It then outputs the inductive invariant $\eta_\sigma := \eta[s_{\loc,i,j} \leftarrow \sigma(s_{\loc,i,j})].$

\begin{example}
	In Example~\ref{ex:running}, we mentioned that our goal is to prove that the return value of $\func{sum}$ is less than $0.5 \cdot n^2 + 0.5 \cdot n + 1$, i.e.~we want to obtain
	 $0.5 \cdot \bar{n}^2 + 0.5 \cdot \bar{n} + 1 -r > 0 ~~~~~~(*)$ 
	 
	\noindent at the endpoint label $9$ of $\func{sum}$. To do so, our algorithm calls a QCLP-solver over the system of quadratic equalities obtained in Example~\ref{ex:step3}, with the objective of minimizing the Euclidean distance between the coefficients synthesized for $\eta(9)$ and those of $(*)$. The QCLP-solver obtains a solution $\sigma$ (i.e.~a valuation to the new unknown $s-$, $t-$ and $\epsilon-$variables), such that $\eta(9)[s_{9, i, j} \leftarrow \sigma(s_{9, i, j})] = 0.5 \cdot \bar{n}^2 + 0.5 \cdot \bar{n} + 1 - r > 0$, hence proving the desired invariant. The complete solution is provided in Appendix~\ref{app:exp-run}.
\end{example}

\begin{remark}[Form of the Objective Function]
	At first sight, the objective functions considered above might seem bizarre, given that they are functions of the unknown $s$-variables, i.e.~the coefficients of the invariant which should be synthesized by the algorithm. In our view, this is a useful formulation. In many cases, the goal of a verification process is to prove that a certain desired invariant $\inv(\loc)$ holds at a specific point $\loc$ of the program. This goal can be specified as an objective function over the $s$-variables. However, it does not simplify the invariant generation problem, because although $\inv(\loc)$ is given, in order to prove that it is an invariant, we have to find an inductive invariant for every other point of the program, too.
\end{remark}

\begin{theorem}[Weak Invariant Synthesis]
	Given a non-recursive program $P$, a pre-condition $\precond$ that satisfies the compactness condition and a linear/quadratic objective function $\textsf{obj}$, the $\weakInvAlgo$ algorithm reduces the Weak Invariant Synthesis problem to QCLP in \emph{polynomial time}. This reduction is sound and semi-complete.
\end{theorem}

	\section{Invariants for Recursive Programs} \label{sec:algo-recursive}

We extend our algorithms to handle recursion. Recall that the only differences between recursive and non-recursive inductive invariants are (i)~presence of function-call statements in recursive programs, (ii)~presence of post-conditions, and (iii)~the post-condition consecution requirement. We expect an invariant generation algorithm for recursive programs to also synthesize a post-condition for every function.  
	
\begin{figure}
	\begin{center}
	\begin{minipage}{0.5\linewidth}
		\lstset{tabsize=4}
		\lstset{language=prog}
		\begin{lstlisting}[mathescape,basicstyle=\small,aboveskip=0pt,belowskip=0pt]
  $\func{rsum}$($n$) {
1:	if $n \leq 0$ then
2:		return $n$
  	else
3:		$m$ := $n - 1$;
4:		$s$ := $\func{rsum}$($m$);
5:		if $\star$ then
6:			$s$ := $s+n$
7:      else skip fi;
8:		return $s$
9:  fi }
		\end{lstlisting}
	\end{minipage}
	\end{center}
\caption{A recursive non-deterministic summation program}
\label{fig:rec-example}
\end{figure}
	
\begin{example}	
	Consider the program in Figure~\ref{fig:rec-example}, which is a recursive variant of the non-deterministic summation program of Figure~\ref{fig:running}. We use this program to illustrate our approach for handling recursion.	
\end{example}

\paragraph{The $\recursiveStrongInvAlgo$ and $\recursiveWeakInvAlgo$ Algorithms} Our algorithm for Strong (resp. Weak) Invariant Synthesis over a recursive program $P$ is called $\recursiveStrongInvAlgo$ (resp. $\recursiveWeakInvAlgo$). It takes the same inputs as in the non-recursive case, except that the input program $P$ can now be recursive. It performs the same steps as in its non-recursive counterpart, except that the following additional actions are taken in Steps 1 and 2:

\smallskip\noindent\textbf{Step 1.a) Setting up a template for the post-condition.}
Let $\hat{\monomials}_d^f = \{ \hat{m}_1, \hat{m}_2, \ldots, \hat{m}_{\hat{r}} \}$ be the set of all monomials of degree at most $d$ over $\{ \ret^f, \bar{v}_1, \ldots, \bar{v}_n \}$.
The algorithm generates an additional template $\mu(f) := \bigwedge_{i=1}^n \varphi_{f, i}$ where each $\varphi_{f, i}$ is of the form $\varphi_{f, i} := \left( \sum_{j=1}^{\hat{r}} s_{f, i, j} \cdot \hat{m}_j > 0 \right)$ where the $s_{f, i, j}$'s are additional new $s$-variables. Intuitively, our goal is to synthesize the right value for $s$-variables such that $(\mu, \eta)$ becomes a recursive inductive invariant. As a consequence, $\mu$ will be a post-condition and $\eta$ a valid invariant.

\begin{example}
	Consider the program in Figure~\ref{fig:rec-example} and assume that each desired invariant/post-condition consists of a single quadratic inequality. The algorithm generates a template $\mu(\func{rsum})$ for the post-condition of \func{rsum}. By definition, such a post-condition can only depend on $\bar{n}$, i.e.~the value passed for the parameter $n$ when \func{rsum} is called, and the return value 
	$r := \ret^\func{rsum}$. Hence, the algorithm generates the following template:
	$
	\mu(\func{rsum}) := s_{\func{rsum}, 1, 1} + s_{\func{rsum}, 1, 2} \cdot \bar{n} + s_{\func{rsum}, 1, 3} \cdot r + s_{\func{rsum}, 1, 4} \cdot \bar{n}^2 + s_{\func{rsum}, 1, 5} \cdot \bar{n} \cdot r + s_{\func{rsum}, 1, 6} \cdot r^2 > 0
	$
\end{example}

\noindent\textbf{Step 2.a) Setting up constraint pairs at function-call statements.} For every transition $e = (\loc, \perp, \loc')$ where $\loc$ is a function-call statement of the form $v_0 := f'(v_1, \ldots, v_n)$ calling a function with header $f'(v'_1, \ldots, v'_n)$, and every polynomial $g$ for which $g>0$ appears in $\eta(\loc')$, the algorithm defines a new program variable $v_0^*$ and adds the following constraint pair to $\Lambda_e$:
$$
\begin{small}
\begin{pmatrix}
 \makecell{\precond(\loc) \wedge \eta(\loc) \wedge \precond(\lin^{f'}) [v'_i \leftarrow v_i, \bar{v}'_i \leftarrow v_i] \wedge\\ \mu(f')[\ret^{f'} \leftarrow v_0^*, \bar{v}'_i \leftarrow v_i] \wedge \precond(\loc')[v_0 \leftarrow v_0^*]} ~,~  g[v_0 \leftarrow v_0^*] 
\end{pmatrix}\end{small}
,
$$

\noindent in which $\phi[x\leftarrow y]$ is the result of replacing every occurrence of $x$ in $\phi$ with a $y$. Intuitively, $v_0^*$ models the value of $v_0$ after the function call (equivalently the return value of $f'$)\footnote{Note that $v_0$ is the only variable in $f$ whose value might change after the call to $f'$. Hence, we need to distinguish between the initial value of $v_0$ and its value after the execution of $f'$, which is denoted by $v_0^*$.}. The constraint pair above encodes the consecution requirement at function-call labels, i.e.~it simply requires every valid \emph{abstract} path that satisfies the invariant at $\loc$ to satisfy it at $\loc'$, too. Note that a valid abstract path must satisfy the post-condition and all the pre-conditions.

\begin{example}
	Consider the transition $e = (4, \perp, 5)$ in Figure~\ref{fig:rec-example}. The algorithm computes the following constraint and adds it to $\Lambda_e$:
	$$
	\begin{small}
	\begin{pmatrix}
	\makecell{\precond(4) \wedge \eta(4) \wedge \precond(1) [n \leftarrow m, \bar{n} \leftarrow m] \wedge\\ \mu(\func{rsum})[\ret^{\func{rsum}} \leftarrow s^*, \bar{n} \leftarrow m] \wedge \precond(5)[s \leftarrow s^*]} ~,~ \eta(5)[s \leftarrow s^*] 
	\end{pmatrix}
	\end{small}
	.
	$$
We now explain this constraint in detail. The purpose of this constraint is to enforce the consecution property in the transition $e$ from label $4$ to label $5$. Recall that the consecution property requires that for every \emph{valid} unit-length \emph{abstract} path $\pi = \langle (\func{rsum}, 4, \nu_4), (\func{rsum}, 5, \nu_5) \rangle$, we have $\nu_4 \satisfies \eta(4) \Rightarrow \nu_5 \satisfies \eta(5).$ Since the variable $s$ is updated in line $4$, we use $s$ to denote its value before execution of the recursive call and $s^*$ to model its value after the function call. Hence, $\nu_5 \satisfies \eta(5)$ can be simply rewritten as $\eta(5)[s \leftarrow s^*]$ (the second component of the above constraint). On the other hand, the first component of the constraint should encode the properties that (a)~$\nu_4 \satisfies \eta(4)$ and (b)~$\pi$ is a valid abstract path. The property (a) is ensured by including $\eta(4)$ in the first component of the constraint. Similarly, (b) is encoded as follows:
\begin{compactitem}
	\item $\precond(4)$ encodes the requirement $\nu_4 \satisfies \precond(4).$
	
	\item $\precond(1)[n \leftarrow m, \bar{n} \leftarrow m]$ encodes the requirement that the function $\func{rsum}$ can be called using the parameter $m$, i.e.~that $m$ satisfies the pre-condition of $1 = \lin^\func{rsum}.$
	
	\item $\mu(\func{rsum})[\ret^{\func{rsum}} \leftarrow s^*, \bar{n} \leftarrow m]$ checks that the call to $\func{rsum}$ is abstracted correctly, i.e.~that the value $s^*$ returned by $\func{rsum}$ respects the post-condition $\mu(\func{rsum})$. 
	
	\item $\precond(5)[s \leftarrow s^*]$ encodes the requirement that the program should be able to continue its execution from point $5$ with the new value of $s$, or equivalently $\nu_5 \satisfies \precond(5).$
\end{compactitem}
\end{example}

\noindent\textbf{Step 2.b) Setting up constraint pairs for post-condition consecution.} For each transition $e = (\loc, \alpha, \loc')$ where $\loc$ is a \textbf{return} statement and $\loc' = \lout^f$ for some program function $f$, the algorithm generates the following constraint pairs:
\begin{compactitem}
	\item 
	For every polynomial $g$ such that $g>0$ appears in $\mu(f)$, the algorithm adds the condition pair $(\precond(\loc) \wedge \eta(\loc) \wedge (\precond(\loc') \circ \alpha), g \circ \alpha)$ to $\Lambda_e$.
\end{compactitem}
These constraints encode post-condition consecution.

\begin{example}
	Consider transition $e = (2, \ret^\func{rsum} \leftarrow n, 9)$ in the program of Figure~\ref{fig:rec-example}. The algorithm generates the following constraint and adds it to $\Lambda_e:$
	$$
	\begin{small}
	\begin{array}{l}
	\left (
	\precond(2) \wedge \eta(2) \wedge \precond(9)[\ret^\func{rsum} \leftarrow n] ~,~  \mu(\func{rsum})[\ret^\func{rsum} \leftarrow n] \right) =
	\end{array}
	\end{small}
	$$
	$$
	\begin{small}
	\begin{array}{l}
	(
	\precond(2) \wedge \eta(2) \wedge \precond(9)[\ret^\func{rsum} \leftarrow n]~,~ s_{\func{rsum}, 1, 1} + s_{\func{rsum}, 1, 2} \cdot \bar{n} \\ + s_{\func{rsum}, 1, 3} \cdot n + s_{\func{rsum}, 1, 4} \cdot \bar{n}^2 + s_{\func{rsum}, 1, 5} \cdot \bar{n} \cdot n + s_{\func{rsum}, 1, 6} \cdot n^2 > 0 ).
	\end{array}
	\end{small}
	$$
	This enforces the post-condition consecution requirement, i.e.~that in every valid execution step going from line $2$ to line $9$, the post-condition $\mu(\func{rsum})$ holds. 

\end{example}

The soundness, completeness and complexity arguments carry over from the non-recursive case. 

\begin{theorem}[Recursive Strong Invariant Synthesis] \label{thm:recstrong}
	Given a recursive program $P$ and a pre-condition $\precond$ that satisfies the compactness condition, the $\recursiveStrongInvAlgo$ algorithm solves the Strong Invariant Synthesis problem in subexponential time. This solution is sound and semi-complete.
\end{theorem}

\begin{theorem}[Recursive Weak Invariant Synthesis] \label{thm:recweak}
	Given a recursive program $P$, a pre-condition $\precond$ that satisfies the compactness condition and a linear/quadratic objective function $\textsf{obj}$, the $\recursiveWeakInvAlgo$ algorithm reduces the Weak Invariant Synthesis Problem to QCLP/QCQP in \emph{polynomial time}. This reduction is sound and semi-complete.
\end{theorem}

	\newcommand{\fail}{\texttt{Failed}}
\newcommand{\timeout}{\texttt{Timed Out}}
\newcommand{\notapp}{\texttt{Not Applicable}}

\section{Experimental Results} \label{sec:exp}

\paragraph{Implementation} We implemented our algorithms for weak invariant generation in Java and used LOQO~\cite{loqo} for solving the QCLPs. All results were obtained on an Intel Core i5-7200U machine with 6 GB of RAM, running Ubuntu 18.04.

\paragraph{Previous Methods} We compare our approach against five previous methods, including three widely-used tools, namely ICRA~\cite{DBLP:journals/pacmpl/KincaidCBR18}, SeaHorn~\cite{seahorn}, and UAutomizer~\cite{uautomizer}, a state-of-the-art method using hypergeometric sequences~\cite{HumenbergerJK17}, and our own implementation of the previous method that provides completeness guarantees for polynomial invariants~\cite{kapur2004automatically}. 

\paragraph{Technical Parameters}
In our experiments, we set $n$ to be the maximal number of desired inequalities given at the same label, i.e.~we used the smallest possible number of conjuncts needed to represent the desired assertions.  Similarly, we let $d$ be the highest degree among desired inequalities in the input and $\Upsilon = d.$ Alternatively, our algorithm can be run iteratively, increasing the values of $d$ and $\Upsilon$ in a diagonal fashion, until the desired invariant is found. Moreover, we did not bound our variables using pre-conditions. 

\paragraph{Solver Errors} Ensuring stability of the QCLP solver is an orthogonal problem. However, to gain confidence that the soundness of our approach is not compromised by potentially cascading numerical errors in our solver, we checked each output using infinite-precision arithmetic, by plugging it back into Equation ($\dagger$) in Step~3 to make sure that (i)~every synthesized strict inequality has a positivity witness of $10^{-9}$ or larger, and (ii)~every instance of $(\dagger)$ corresponding to an equality $g = 0$ holds within an error margin of $10^{-9}.$

\paragraph{Non-recursive Results} We used the benchmarks in~\cite{benchmarks}, which contain programs, pre-conditions, and the desired post-conditions and assertions (invariants at a few labels) that are needed for their verification. The problem is to find an inductive invariant that proves the given post-conditions and assertions. We ignored benchmarks that contained non-polynomial assignments or pre-conditions. The results are summarized in Table~\ref{tab:res-non}. Our algorithm is not complete for \emph{non-strict} invariants (Remark~\ref{rem:no-str}), but it successfully generated all the desired invariants for these benchmarks.

\begin{table*}
	\begin{footnotesize}
		\setlength\extrarowheight{0pt}
	\begin{tabular}{|c|c|c|c|c|c||c|c|c|c|c|}
		\hline
		\textbf{Benchmark}	&$\mathbf{n}$	&$\mathbf{d}$	&$\vert \pvars \vert$ & $\vert \mathbf{S} \vert$	& \textbf{Ours} & \textbf{ICRA} & \textbf{SeaHorn} & \textbf{\cite{HumenbergerJK17}} & \textbf{UAutomizer} & \textbf{\cite{kapur2004automatically} using Z3}  \\ \hline \hline
		cohendiv	&3	&2	&6	& 17391 &15.2 & 0.7 & 0.1 & \notapp & 3.3 & \timeout\\ \hline
		divbin	&3	&2	&5	&  18351 &5.4 & \fail & \timeout & 0.2 & \fail & \timeout\\ \hline
		hard	&3	&2	&6	& 24975 &28.0 & \fail & \fail & 0.4 & \fail & \timeout\\ \hline
		mannadiv	&3	&2	&5	& 16245 &18.2 & \fail & 0.1 & 0.1 & \timeout & \timeout\\ \hline
		wensely	&2	&2	&7	& 18874  & 20.1 & \fail & \fail & 0.1 & \fail & \timeout\\ \hline
		sqrt	&2	&2	&4	& 4072 &5.8 & 0.8 & \fail & 0.1 & \timeout& \timeout\\ \hline
		dijkstra	&2	&2	&5	& 10156 &12.8 & \fail & \fail & \notapp & \fail& \timeout\\ \hline
		z3sqrt	&2	&2	&6	& 9404 & 12.9 & 0.5 & 0.1 & \notapp & \fail & \timeout\\ \hline
		freire1	&2	&2	&3	& 2432 & 26.5 & 0.6 & \fail & 0.1 & \fail & \timeout\\ \hline
		freire2	&2	&3	&4	& 9708 & 10.7 & 1.1 & \fail & 0.1 & \fail & \timeout\\ \hline
		euclidex1	&2	&2	&11	& 45756 &97.5 & \fail & \fail & \notapp & \timeout & \timeout\\ \hline
		euclidex2	&2	&2	&8	& 22468 &39.3 & \fail & \fail & 0.4 & \timeout& \timeout \\ \hline
		euclidex3	&2	&2	&13	& 72762 &203.1 & \fail & \fail & \notapp & \timeout& \timeout \\ \hline
		lcm1	&2	&2	&6	& 13361 & 17.9 & 0.8 & 0.1 & \notapp & 3.7& \timeout \\ \hline
		lcm2	&2	&2	&6	& 12517 & 18.7 & 0.8 & 0.1 & 0.1 & 3.2& \timeout\\ \hline
		prodbin	&2	&2	&5	& 10096 & 12.1 & \fail & \fail & \notapp & \timeout& \timeout \\ \hline
		prod4br	&2	&2	&6	& 21064 & 43.2 & \fail & \fail & \notapp & \timeout& \timeout \\ \hline
		cohencu	&2	&3	&5	& 16664 & 11.8 & 0.6 & \fail & 0.1 & \timeout& \timeout \\ \hline
		petter	&1	&2	&3	& 1080 & 20.4 & 0.5 & 0.1 & 0.1 & 2.7& \timeout\\ \hline
	\end{tabular}
\end{footnotesize}
	\caption{Experimental results over the benchmarks of~\cite{benchmarks}. $\vert \pvars \vert$ is number of program variables and $\vert \mathbf{S} \vert$ is size of the quadratic system, i.e.~number of constraints in Step~3. Runtimes are reported in seconds. We set a time-limit of 1 hour. }
	\label{tab:res-non}
\end{table*}

\paragraph{Recursive Results} Recursive results are shown in Table~\ref{tab:res-rec}. Our recursive benchmarks can be divided in two categories:
\begin{compactitem}

\item \emph{Reinforcement Learning.} We ran our approach on three programs from~\cite{qlearn} which are used for safety verification of reinforcement learning applications in cyber-physical systems such as Segway transporters. In these examples, the desired partial invariants are linear. However, the programs themselves contain polynomial assignments and conditions of degree $4$. Thus, approaches for linear invariant generation, such as~\cite{sankaranarayanan2004constraint22}, are not applicable. 

\item \emph{Classical Examples.} We considered Figure~\ref{fig:rec-example}, and its extensions to sums of squares and cubes,  to show that our algorithm is able to synthesize invariants of higher degrees. We also considered a program that recursively computes the largest power of $2$ that is no more than a given bound $x$, showing that our algorithm can handle recursive invariants with more than one assertion at each label. Finally, we generated invariants for an implementation of the Merge Sort algorithm that counts number of inversions in a sequence~\cite{cormen2009introduction}. See Appendix~\ref{app:rec-exp} for details. 
\end{compactitem}

\begin{table*}
	\begin{footnotesize}
		\setlength\extrarowheight{0pt}
	\begin{tabular}{|C{1.7cm}|C{1.5cm}|c|c|c|c|c||c|c|c|c|c|}
		\cline{2-12}
		\multicolumn{1}{c|}{}& \textbf{Benchmark}	&$\mathbf{n}$	&$\mathbf{d}$	&$\vert \pvars \vert$ & $\vert \mathbf{S} \vert$	& \textbf{Ours} & \textbf{ICRA} & \textbf{SeaHorn} & \textbf{\cite{HumenbergerJK17}} & \textbf{UAutomizer} & \textbf{\cite{kapur2004automatically} using Z3}\\ \hhline{~===========}\cline{1-1}
		\multirow{3}{*}{\makecell{Reinforcement\\Learning\\ \cite{qlearn}}}& inverted-pendulum & 1 & 3 & 7 & 9951 & 496.1 & \fail & \fail & \notapp & \fail & \timeout \\ \cline{2-12}
		& strict-inverted-pendulum & 4 & 2 & 7 & 14390 & 587.8 & 11.5 & \fail & \notapp & \fail & \timeout \\ \cline{2-12}
		& oscillator & 1 & 2 & 7 & 3552 & 39.7 & \fail & \fail & \notapp & \fail & \timeout  \\\hline \hline 
		\multirow{5}{*}{\makecell{Classical\\Examples\\ (Appendix~\ref{app:rec-exp})}}& recursive-sum & 1 & 2 & 3 & 1700 & 10.9 & 0.6 & \fail & \notapp & \timeout & \timeout \\ \cline{2-12}
		& recursive-square-sum & 1 & 3 & 3 & 1121 & 17.4 & \fail & \fail & \notapp & \fail & \timeout \\ \cline{2-12}
		& recursive-cube-sum & 1 & 4 & 3 & 15840 &221.2 & \fail & \fail & \notapp & \fail & \timeout \\ \cline{2-12}
		& pw2 & 2 & 1 & 3 & 430 & 5.4 & 0.7 & 0.1 & \notapp & \fail & \timeout \\ \cline{2-12}
		& merge-sort & 1 & 2 & 13 & 33002 & 78.1 & \fail & \fail & \notapp & \fail & \timeout \\ \hline
	\end{tabular}
\end{footnotesize}
	\caption{Experimental results over recursive programs. The results are reported in the same manner as in Table~\ref{tab:res-non}.}
	\label{tab:res-rec}
\end{table*}

\paragraph{Runtimes}
	Our runtimes over these benchmarks are typically under a minute, while the maximum runtime is close to 10 minutes. This shows that our approach is applicable in practice and does not suffer from the same impracticalities as~\cite{grigoriev1988solving}, which would take years on problems of this size~\cite{comparison}. 
	
\paragraph{Comparison with Complete Approaches} Almost none of the previous complete approaches are applicable to our benchmarks due to the existence of non-linear assignments and also because the desired invariants are polynomial \emph{inequalities} (See Table~\ref{tab:comp}). The only previous complete approach that handles polynomial programs and polynomial inequalities in invariants is~\cite{kapur2004automatically}. However, it relies on quantifier elimination and is extremely inefficient. We confirmed this point experimentally. We generated the constraints of~\cite{kapur2004automatically} for our benchmarks
and used state-of-the-art quantifier elimination / SMT solver tools (Mathematica~\cite{Mathematica}, QEPCAD~\cite{brwonq} and Z3~\cite{z3}) to solve them. In all cases, the solver either did not terminate, even when we increased the timeout to \emph{12 hours}, or returned with failure. This was the case even for our simple running example (Figure~\ref{fig:running}).

\paragraph{Comparison with Incomplete Approaches}
	As is evident in Tables~\ref{tab:res-non} and \ref{tab:res-rec}, our approach is  slower than previous sound methods that do not provide any completeness guarantee. However, it is able to handle a strictly more general set of benchmarks. Specifically, there are several benchmarks, especially among the recursive programs, where our approach was the only one that could successfully prove the desired assertions. Hence, there is currently a trade-off between accuracy (completeness guarantees) and efficiency. 
	While the semi-completeness guarantee is a key novelty of our approach, 
	we expect that advancements in quadratic programming, which is an active research topic in optimization, will narrow the runtime gap. 
	
	\paragraph{Generality and Types of Invariants} As shown in Tables~\ref{tab:res-non} and~\ref{tab:res-rec}, our approach is able to synthesize polynomial invariants of various degrees for a variety of benchmarks. None of the previous tools can handle all the benchmarks in Tables~\ref{tab:res-non} and~\ref{tab:res-rec}, and there are several instances where our approach is the only successful method. Moreover, we can also successfully synthesize invariants containing polynomial equalities. See Appendix~\ref{app:equality} for a detailed demonstration. This being said, the power of our approach becomes much more apparent when we consider recursive programs (Table~\ref{tab:res-rec}). On our recursive benchmarks, every other method fails in almost all cases. Additionally, our approach is also able to synthesize invariants for two or more functions that recursively call each other. See Appendix~\ref{app:twofunc} for a detailed example of this. Finally, in Appendix~\ref{app:frac}, we show a classical program that approximates an irrational number using its continued fraction representation. This example requires invariants of degree $5$, which are beyond the reach of previous methods. We manually tried all the methods in Table~\ref{tab:comp} and every one of them was either not applicable to this example or failed to synthesize the required invariants. In contrast, our approach could easily handle this program.
	\section{Conclusion}
We presented a subexponential sound and semi-complete method to generate polynomial invariants for programs with polynomial updates. On the practical side, we demonstrated how to generate such invariants using QCLP.
Previous methods were either extremely inefficient or lacked completeness guarantees. An interesting, but non-trivial, direction of future work is to exploit special structural properties of CFGs, such as sparsity and low treewidth~\cite{thorup1998all}, to speed up the solution of our QCLP instances. Such techniques have previously been applied for solving linear programs~\cite{yen2015sparse} and systems of linear equations~\cite{fomin2018fully}, but not QCLPs. 

\newpage
\paragraph{Acknowledgments} The research was partially supported by Austrian Science Fund (FWF) Grant
No.~NFN S11407-N23 (RiSE/SHiNE), Vienna
Science and Technology Fund (WWTF) Project ICT15-003, National Natural Science foundation of China (NSFC) Grant No.~61802254, Facebook PhD Fellowship Program, and DOC Fellowship No.~24956 of the Austrian Academy of Sciences (\"OAW).
	
	\newpage
	\bibliography{refs}
	
	\clearpage
	\appendix
\section{Detailed Syntax} \label{app:syntax}

\paragraph{Polynomial Arithmetic Expressions} A \emph{polynomial arithmetic expression} $\mathfrak{e}$ over $\pvars$ is an expression built from the variables in $\pvars$, real constants, and the arithmetic operations of addition, subtraction and multiplication. 

\paragraph{Propositional Polynomial Predicates} A \emph{propositional polynomial predicate} is a propositional formula built from (i)~\emph{atomic assertions} of the form $\mathfrak{e}_1 \bowtie \mathfrak{e}_2$, where $\mathfrak{e}_1$ and $\mathfrak{e}_2$ are polynomial arithmetic expressions, and $\bowtie \quad\!\!\! \in \{ <, \leq, \geq, > \}$ and (ii)~propositional connectives $\vee$, $\wedge$ and $\neg$. The satisfaction relation $\satisfies$ between a valuation $\nu$ and a propositional polynomial predicate $\phi$ is defined in the natural way, i.e.~by substituting the variables with their values in $\nu$ and evaluating the resulting boolean expression.

\paragraph{Detailed Grammar} Figure~\ref{fig:syntax-complete} provides a more detailed grammar specifying the syntax of non-deterministic recursive programs with polynomial assignments and guards.

\begin{figure}[!h]
	\scalebox{0.9}{
	$\begin{array}{rrl}
	
	\obj{prog} & ::= &  \obj{func}\\
	&& \mid \obj{func}~\obj{prog}\\
	
	\obj{func} & ::= & \obj{fname}~\kw{$($}~\obj{varlist}~\kw{$)$}~\kw{$\{$}~\obj{stmtlist}~\kw{$\}$} \\ 
	
	\obj{varlist} & ::= & \obj{var}\\
	&& \mid \obj{var}~\kw{$,$}~\obj{varlist}\\
	
	\obj{stmtlist} & ::= & \obj{stmt} \\
	&& \mid \obj{stmt}~\kw{$;$}~\obj{stmtlist}\\
	
	\obj{stmt} & ::= & \kw{skip}\\
	&& \mid \obj{var}~\kw{$:=$}~\obj{expr}\\
	&& \mid \kw{if}~\obj{bexpr}~\kw{then}~\obj{stmtlist}~\kw{else}~\obj{stmtlist}~\kw{fi}\\
	&& \mid \kw{if}~\kw{$\star$}~\kw{then}~\obj{stmtlist}~\kw{else}~\obj{stmtlist}~\kw{fi}\\
	&& \mid \kw{while}~\obj{bexpr}~\kw{do}~\obj{stmtlist}~\kw{od}\\
	&& \mid \obj{var}:=\obj{fname}~\kw{$($}~\obj{varlist}~\kw{$)$}\\
	&& \mid \kw{return}~\obj{expr}
	\\
	
	\obj{bexpr} & ::= & \obj{literal}\\
	&& \mid \kw{$\neg$}~\obj{bexpr} \\
	&& \mid \obj{bexpr}~\kw{$\vee$}~\obj{bexpr}\\
	&& \mid \obj{bexpr}~\kw{$\wedge$}~\obj{bexpr}\\
	
	\obj{literal} & ::= & \obj{expr}~\kw{$<$}~\obj{expr}\\
	&& \mid \obj{expr}~\kw{$\leq$}~\obj{expr}\\
	&& \mid \obj{expr}~\kw{$\geq$}~\obj{expr}\\
	&& \mid \obj{expr}~\kw{$>$}~\obj{expr}\\

	\obj{expr} & ::= & \obj{var}\\
	&& \mid \obj{constant} \\
	&& \mid \obj{expr}~\kw{$+$}~\obj{expr}\\
	&& \mid \obj{expr}~\kw{$-$}~\obj{expr}\\
	&& \mid \obj{expr}~\kw{$*$}~\obj{expr}\\
	
	\end{array}$}
	\caption{Detailed Syntax of Non-deterministic Recursive Programs}
	\label{fig:syntax-complete}
\end{figure}

Below, we intuitively explain some aspects of the syntax: 
\begin{compactitem}
	\item \emph{Variables and Function Names.} Expressions $\obj{var}$ (resp. $\obj{fname}$) range over the set $\pvars$ (resp.~$\fnames$).
	\item \emph{Arithmetic and Boolean Expressions.} Expressions $\obj{expr}$ range over all polynomial arithmetic expressions over program variables. Similarly, expressions $\obj{bexpr}$ range over propositional polynomial predicates.
	\item \emph{Statements.} A statement can be one of the following:
	\begin{compactitem}
		\item A special \kw{skip} statement which does not do anything,
		\item An assignment statement ($\obj{var}$~\kw{$:=$}~$\obj{expr}$), 
		\item A conditional branch (\kw{if}~$\obj{bexpr}$) in which the $\obj{bexpr}$ serves as the branching condition;
		\item A non-deterministic branch (\kw{if~$\star$}),
		\item A while-loop (\kw{while}~$\obj{bexpr}$) in which the $\obj{bexpr}$ serves as the loop guard;
		\item A function call statement ($\obj{var}:=\obj{fname}~\kw{$($}~\obj{varlist}~\kw{$)$}$) which calls the function specified by $\obj{fname}$ using the parameters specified in the $\obj{varlist}$ and assigns the resulting returned value to the variable on its left hand side;
		\item A return statement (\kw{return}~$\obj{expr}$) that ends the current function and returns the value of the expression~$\obj{expr}$ and the control to the parent function or ends the program if there is no parent function.
	\end{compactitem}
	\item \emph{Programs and Functions.} A program is simply a list of functions. Each function has a name, a set of parameters and a body. The function body is a sequence of statements. We assume that there is a distinguished function $\mainFunction$ that is the starting point of the program.
\end{compactitem}

\paragraph{Syntactic Assumptions} We assume that each function in $\fnames$ is defined exactly once in the program, function headers do not contain duplicate variables, and each function call statement provides exactly as many parameters as defined in the header of the function that is being called. Moreover, we assume that no variable appears in both sides of a function call statement. 

\paragraph{Simple vs.~Recursive Programs} We call a program \emph{simple}, or \emph{non-recursive}, if it contains only one function and no function call statements. Otherwise, we say that the program is \emph{recursive}.

\section{Detailed Semantics} \label{app:semantics}

\paragraph{Valuations} A \emph{valuation} over a set $W \subseteq \pvars$ of variables is a function $\nu : W \rightarrow \mathbb{R}$ that assigns a real value to each variable in $W.$ We denote the set of all valuations on $W$ by $\vals{W}$. We sometimes use a valuation $\nu$ over a set $W' \subset W$ of variables as a valuation over $W$. In such cases, we assume that $\nu(w) = 0$ for every $w \in W \setminus W'$. Given a valuation $\nu$, a variable $v$ and $x \in \mathbb{R}$, we write $\nu[v \leftarrow x]$ to denote a valuation $\nu'$ such that $\nu'(v) = x$ and $\nu'$ agrees with $\nu$ for every other variable.

\paragraph{Notation} We define $\pvars_*^f := \{ \ret^f, v_1, \ldots, v_n, \bar{v}_1, \ldots, \bar{v}_n \}$ and let $\pvars^f$ be the set containing all members of $\pvars_*^f$, as well as any variable that appears somewhere in the body of the function $f$.  W.l.o.g.~we assume that the $\pvars^f$'s are pairwise disjoint. Moreover,  we write $\vals{f}$ as a shorthand for $\vals{\pvars^f}$. In other words, $\vals{f}$ is the set of all valuations over the variables that appear in $f$, including its header, its body and its new variables.  Similarly, we define  $\labels^f$ as the set of labels that occur in $f$.

\paragraph{Configurations} A \emph{stack element} $\xi$ is a tuple $(f, \loc, \nu)$ where $f \in \fnames$ is a function and $\loc \in \labels^f$ and $\nu \in \vals{f}$ are respectively a label and a valuation in $f$. A configuration $\kappa = \langle \xi_i\rangle_{i=0}^n$ is a finite sequence of stack elements. 

\paragraph{Notation} Given a configuration $\kappa$ and a stack element $\xi$, we write $\kappa \cdot \xi$ to denote the configuration obtained by adding $\xi$ to the end of $\kappa$. Also, we define $\kappa^{-i}$ as the sequence obtained by removing the last $i$ stack elements of $\kappa$.

\paragraph{Runs} A \emph{run} is an infinite sequence of configurations that starts at the first label of $\mainFunction$ and follows the requirements of the CFG. Intuitively, a run models the sequence of configurations that are met in an execution of the program.

\paragraph{Formal Definition of Runs} Given a program $P$ and its CFG $(\fnames, \labels, \rightarrow)$, a \emph{run} is a sequence $\rho = \{\kappa_i\}_{i=0}^\infty$ of configurations such that:
\begin{compactitem}
	\item $\kappa_0 = \langle (\mainFunction, \lin^\mainFunction, \nu) \rangle$ for some valuation $\nu \in \vals{\mainFunction}$. Intuitively, a run begins from the $\mainFunction$ function.
	\item If $\vert \kappa_i \vert = 0$, then $\vert \kappa_{i+1} \vert = 0$, too. Informally, this case corresponds to when the program has already terminated.
	\item Let $\xi = (f, \loc, \nu)$ be the last stack element in $\kappa_i$. Then, $\kappa_{i+1}$ should satisfy one of the following rules:
	\begin{compactitem}
		\item[(a)] $\loc \in \locs_a$ and $(\loc, \alpha, \loc') \in \rightarrow$ and $\kappa_{i+1} = \kappa_i^{-1} \cdot (f, \loc', \alpha(\nu))$.
		\item[(b)] $\loc \in \locs_b$ and $(\loc, \phi, \loc') \in \rightarrow$ where $\phi$ is a predicate such that $\nu \satisfies \phi$ and $\kappa_{i+1} = \kappa_i^{-1} \cdot (f, \loc', \nu)$.
		\item[(c)] $\loc \in \locs_c$, the statement corresponding to $\loc$ is the function call $v_0 := f'(v_1, v_2, \ldots, v_n)$,  the header of the function $f'$ is $f'(v'_1, v'_2, \ldots, v'_n)$, and $\kappa_{i+1} = \kappa_i \cdot (f', \lin^{f'}, \nu')$ where
		$$
		\nu'(x) = 
		\left\{\begin{matrix}
		\nu(v_i) & &  ~~x \in \{ v'_i, \bar{v}'_i \}\\ 
		0 & &\text{otherwise}
		\end{matrix}\right..
		$$
		Intuitively, this corresponds to adding the new function to the stack. 
		\item[(d)] $\loc \in \locs_d$ and $(\loc, \star, \loc') \in \rightarrow$ and $\kappa_{i+1} = \kappa_i^{-1} \cdot (f, \loc', \nu)$.
		\item[(e$_1$)] $\loc \in \locs_e$ and $\vert \kappa_i \vert = 1$ and $\vert \kappa_{i+1} \vert = 0$. Informally, this case corresponds to the termination of the program when the $\mainFunction$ function returns and the stack becomes empty.
		\item[(e$_2$)] $\loc \in \locs_e$, $\vert \kappa_i \vert > 1$, $\hat{\xi} = (
		\hat{f}, \hat{\loc}, \hat{\nu})$ is the stack element before $\xi$ in $\kappa_i$, the label $\hat{\loc}$ corresponds to a function call of the form $v_0 := f(v_1, \ldots, v_n)$, $(\hat{\loc}, \perp, \hat{\loc'}) \in \rightarrow$ and $\kappa_{i+1} = \kappa_i^{-2} \cdot (\hat{f}, \hat{\loc'}, \hat{\nu}[v_0 \leftarrow \nu(\ret^f)])$. Informally, this corresponds to returning control from the function $f$ into its parent function $\hat{f}$. 
	\end{compactitem}
\end{compactitem}

\paragraph{Return Assumption} We assume that every execution of a function ends with a \textbf{return} statement. If this is not the case, we can add ``\textbf{return} 0'' to suitable points of the program to obtain an equivalent program that satisfies this condition.

\paragraph{Semi-runs and Paths} A \emph{semi-run} starting at a stack element $\xi = (f, \loc, \nu)$ is a sequence $\varrho = \langle \kappa_i \rangle_{i=0}^\infty$ that satisfies all the conditions of a run, except that it starts with $\kappa_0 = \langle \xi \rangle$. A \emph{path} $\pi = \langle\kappa_i \rangle_{i=0}^n$ of length $n$ is a finite prefix of a semi-run.

\paragraph{Valid Runs} A run $\rho$ is \emph{valid} w.r.t.~a pre-condition $\precond$, if for every stack element $\xi = (f, \loc, \nu)$ appearing in one of its configurations, we have $\nu \satisfies \precond(\loc)$. Valid semi-runs and paths are defined similarly. A stack element is \emph{reachable} if it appears in a valid run.

\paragraph{Abstract Paths} Given a pre-condition $\precond$ and a post-condition $\postcond$, an \emph{abstract path} starting at a stack element $\xi = (f, \loc_0, \nu_0)$ is a sequence $\varpi = \langle \kappa_i = \langle(f, \loc_i, \nu_i)\rangle \rangle_{i=0}^{n}$ such that for all $i<n$, $\kappa_{i+1}$ satisfies either one of the conditions (a), (b) and (d) as in the definition of runs or the following modified (c) condition:
\begin{itemize}
	\item[(c$'$)] $\loc_i \in \locs_c$, i.e.~the statement corresponding to $\loc_i$ is a function call $v_0 := f'(v_1, \ldots, v_n)$ where $f'$ is a function with the header	$f'(v'_1, \ldots, v'_n)$ and $\nu_i \satisfies \precond(\lin^{f'})[v'_i \leftarrow v_i, \bar{v}'_i \leftarrow v_i]$. Moreover, $(\loc_i, \perp, \loc_{i+1}) \in \rightarrow$, the valuation $\nu_{i+1}$ agrees with $\nu_i$ over every variable, except possibly $v_0$, and $\nu_{i+1} \satisfies \postcond(f)[\bar{v}_i' \leftarrow v_i, \ret^f \leftarrow v_0]$. The latter is the result of replacing each occurrence of $\bar{v}_i'$ with its respective $v_i$ and $\ret^f$ with $v_0$ in $\postcond(f)$.
\end{itemize}
An abstract path always remains in the same function and hence each configuration in an abstract path consists of only one stack element. A valid abstract path is defined similarly to a valid path. 

\section{Inductive Assertion Maps and Invariants} \label{app:invariant}

\paragraph{Invariants} Given a program $P$ and a pre-condition $\precond$, an \emph{invariant} is a function $\invariant$ mapping each label $\loc \in \locs^f$ of the program to a conjunctive propositional formula $\inv(\loc) := \bigwedge_{i=0}^m \left( \mathfrak{e}_i > 0 \right)$ over $\pvars^f$, such that for every reachable stack element $(f, \loc, \nu)$, it holds that $\nu \satisfies \inv(\loc)$. 

\paragraph{Inductive Assertion Maps} Given a \emph{non-recursive} program $P$ and a pre-condition $\precond$, an \emph{inductive assertion map} is a function $\ind$ mapping each label $\loc \in \locs^f$ of the program to a conjunctive propositional formula $\ind(\loc) := \bigwedge_{i=0}^m \left( \mathfrak{e}_i > 0 \right)$ over $\pvars^f$, such that the following two conditions hold:
\begin{compactitem}
	\item \emph{Initiation.} For every stack element $\xi = (\mainFunction, \lin^\mainFunction, \nu_0)$, we have $\nu_0 \satisfies \precond(\lin^\mainFunction) \Rightarrow \nu_0 \satisfies \ind(\lin^\mainFunction)$. Intuitively, this means that $\ind(\lin^\mainFunction)$ should be deducible from the pre-condition $\precond(\lin^\mainFunction).$ 
	\item \emph{Consecution.} For every valid unit-length path $\pi = \langle (\mainFunction, \loc_0, \nu_0),\\ (\mainFunction, \loc_1, \nu_1) \rangle$ that starts at $\loc_0$ and ends at $\loc_1$, we have $\nu_0 \satisfies \ind(\loc_0) \Rightarrow \nu_1 \satisfies \ind(\loc_1).$ Intuitively, this condition means that the inductive assertion map cannot be falsified by running a valid step of the execution of the program.
\end{compactitem}

	\begin{lemma} 
		Given a non-recursive program $P$ and a pre-condition $\precond$, every inductive assertion map $\ind$ is an invariant.
\end{lemma}

\begin{proof}
	Consider a valid run $\rho = \langle \kappa_i  \rangle_{i=0}^\infty = \langle \langle (\mainFunction, \loc_i, \nu_i) \rangle \rangle_{i=0}^\infty$ of $P$. Let $\pi = \langle \kappa_i \rangle_{i=0}^n$ be a prefix of $\rho$, which is a valid path of length $n$. We prove that $\nu_n \satisfies \ind(\loc_n)$. Our proof is by induction on $n$. 
		For the base case of $n=0$, we have $\loc_0 = \lin^\mainFunction$. By validity of $\rho$, we have $\nu_0 \satisfies \precond(\lin^\mainFunction)$. Hence, by initiation, $\nu_0 \satisfies \ind(\lin^\mainFunction)$. 
		For the induction step, assuming that $\nu_{n-1} \satisfies \ind(\loc_{n-1})$, we prove that $\nu_{n} \satisfies \ind(\loc_{n}).$ We apply the consecution property to the unit-length valid path $\langle (\mainFunction, \loc_{n-1}, \nu_{n-1}), (\mainFunction, \loc_n, \nu_n) \rangle$, which leads to $\nu_n \satisfies \ind(\loc_n)$. 
	
		Hence, for every reachable stack element $(\mainFunction, \loc, \nu)$, we have $\nu \satisfies \ind(\loc)$, which means $\ind$ is an invariant. \qed
\end{proof}

\paragraph{Recursive Inductive Invariants} Given a recursive program $P$ and a pre-condition $\precond$, a \emph{recursive inductive invariant} is a pair $(\postcond, \ind)$ where $\postcond$ is a post-condition and $\ind$ is a function that maps every label $\loc \in \locs^f$ of the program to a conjunctive propositional formula $\ind(\loc) := \bigwedge_{i=0}^m \left( \mathfrak{e}_i > 0 \right)$, such that the following requirements are met:
\begin{compactitem}
	\item \emph{Initiation.} For every stack element $\xi = (f, \lin^f, \nu_0)$ at  start of a function $f$, we have $\nu_0 \satisfies \precond(\lin^f) \Rightarrow \nu_0 \satisfies \ind(\lin^f).$
	\item \emph{Consecution.} For every valid unit-length \emph{abstract} path $\pi = \langle (f, \loc_0, \nu_0), (f, \loc_1, \nu_1) \rangle$ that starts at $\loc_0 \in \locs^f$ and ends at $\loc_1 \in \locs^f$, we have $\nu_0 \satisfies \ind(\loc_0) \Rightarrow \nu_1 \satisfies \ind(\loc_1).$ 
	
	\item \emph{Post-condition Consecution.} For every valid unit-length \emph{abstract} path $\pi = \langle (f, \loc_0, \nu_0), (f, \lout^f, \nu_1) \rangle$ that starts at $\loc_0 \in \locs^f$ and ends at the endpoint label $\lout^f$, we have $\nu_0 \satisfies \ind(\loc_0) \Rightarrow \nu_1 \satisfies \postcond(f).$
\end{compactitem}

	\begin{lemma}
		Given a recursive program $P$ and a pre-condition $\precond$, if  $(\postcond, \ind)$ is a recursive inductive invariant, then the function $\ind$ is an invariant.
\end{lemma}
\begin{proof}
		Consider an arbitrary valid run $\rho = \langle \kappa_i  \rangle_{i=0}^\infty$ of $P$. Let $\pi = \langle \kappa_i \rangle_{i=0}^n$ be a prefix of $\rho$, which is a valid path of length $n$ and $\xi = (f, \loc, \nu)$ the last stack element of $\kappa_n$. We prove that $\nu \satisfies \ind(\loc)$. Our proof is by induction on $n$.

		For the base case of $n=0$, we have $\loc = \lin^\mainFunction.$ By validity of $\rho$, we have $\nu \satisfies \precond(\lin^\mainFunction)$. Hence, by initiation, $\nu \satisfies \ind(\lin^\mainFunction).$ For the inductive step, we let $\xi' = (f', \loc', \nu')$ be the last stack element in $\kappa_{n-1}$. We prove that $\nu \satisfies \ind(\loc).$ We consider the following cases:
		\begin{itemize}
			\item If $f' = f$, then $\langle (f', \loc', \nu'), (f, \loc, \nu) \rangle$ is a valid abstract path of length $1$. Hence, by consecution, we have $\nu \satisfies \ind(\loc).$
			\item If $f'$ is the parent function of $f$, i.e.~$\loc = \lin^f$ and $\loc'$ is a function-call statement calling $f$, then by validity of $\rho$, we have $\nu \satisfies \precond(\loc)$ and by initiation, we infer $\nu \satisfies \ind(\loc)$.
			\item If $f$ is the parent function of $f'$, i.e.~$\loc' = \lout^{f'}$, then let $\hat{\xi} = (f, \hat{\loc}, \hat{\nu})$ be the last visited stack element in $f$ before $f'$ was called. It is easy to verify that $\hat{\loc}$ is a function-call statement calling $f'$ and $(\hat{\loc}, \perp, \loc) \in \rightarrow$. By post-condition consecution, $\nu' \satisfies \postcond(f')$, hence $\langle (f, \hat{\loc}, \hat{\nu}), (f, \loc, \nu ) \rangle$ is a valid abstract path of length $1$. By the induction hypothesis, we have $\hat{\nu} \satisfies \ind(\hat{\loc})$, hence, by consecution, we deduce $\nu \satisfies \ind(\loc).$
	\end{itemize}
	
		Hence, for every reachable stack element $\xi = (f, \loc, \nu)$, we have $\nu \satisfies \ind(\loc)$ which means $\ind$ is an invariant. \qed
\end{proof}

\section{Mathematical Tools and Lemmas} \label{app:maths}

\subsection{Proof of Corollary~\ref{col:put}} \label{app:col:put}

\noindent\textbf{Corollary~\ref{col:put}.}
	Let $V, g, g_1, \ldots, g_m$ and $\Pi$ be as in Theorem~\ref{thm:putinar}. Then $g(x) > 0$ for all $x \in \Pi$ \emph{if and only if}:
	\begin{equation} \label{eq:colput2}
	g = \epsilon + h_0 + \sum_{i=1}^m h_i \cdot g_i
	\end{equation}
	where $\epsilon >0$ is a real number and each polynomial $h_i$ is the sum of squares of some polynomials in $\mathbb{R}[V]$.

\begin{proof}
		It is obvious that if~(\ref{eq:colput2}) holds, then $g(x)>0$ for all $x \in \Pi$. We prove the other side. Let $g(x) > 0$ for all $x \in \Pi$. Given that $\Pi$ is compact and $g$ continuous, $g(\Pi)$ must also be compact and hence closed. Therefore, $\delta := \inf_{x \in \Pi} g(x) > 0$. Let $\epsilon = \delta / 2$, then $g(x) - \epsilon > 0$ for all $x \in \Pi$. Applying Putinar's Positivstellensatz, i.e. equation (\ref{eq:putinar}), to $g - \epsilon$ leads to the desired result. \qed
\end{proof}

\subsection{Proof of Lemma~\ref{lm:sq}} \label{app:maths2}

In our algorithm, we have to reduce the problem of checking whether a polynomial $h \in \mathbb{R}[V]$ is a sum-of-squares to solving a quadratic system. We now present this reduction in detail. Our reduction is based on the following two well-known theorems:

\begin{theorem}[See \cite{horn1990matrix}, Corollary 7.2.9]
	A polynomial $h \in \mathbb{R}[V]$ of even degree $d$ is a sum-of-squares \emph{if and only if} there exists a $k$-dimensional symmetric positive semi-definite matrix $Q$ such that $h = y^T Q y$, where $k$ is the number of monomials of degree no greater than $d/2$ and $y$ is a column vector consisting of every such monomial.
\end{theorem}

\begin{theorem}[\cite{cholesky,gene1996matrix}]
	A symmetric square matrix $Q$ is positive semi-definite \emph{if and only if} it has a Cholesky decomposition of the form $Q = L L^T$ where $L$ is a lower-triangular matrix with non-negative diagonal entries.
\end{theorem}

Given the two theorems above, our reduction uses the following procedure for generating quadratic equations that are equivalent to the assertion that $h$ is a sum-of-squares:

\paragraph{The Reduction} The algorithm generates the set $\monomials_{\lfloor d/2 \rfloor}$ of monomials of degree at most $\lfloor d/2 \rfloor$ over $V.$ It then orders these monomials arbitrarily into a vector $y$ and symbolically computes the equality 
\begin{equation} \label{eq:sos}
h = y^T L L^T y
\end{equation}
where $L$ is a lower-triangular matrix whose every non-zero entry is a new variable in the system. We call these variables $l$-variables. For every $l_{i,i}$, i.e.~every $l$-variable that appears on the diagonal of $L$, the algorithm adds the constraint $l_{i,i} \geq 0$ to the quadratic system. Then, it translates Equation~\eqref{eq:sos} into quadratic equations over the coefficients of $h$\footnote{These coefficients are called $t$-variables in our algorithm.} and $l$-variables by equating the coefficients of corresponding terms on the two sides of \eqref{eq:sos}. The resulting system encodes the property that $h$ is a sum-of-squares.

\begin{example}
	Let $V = \{a, b\}$ be the set of variables and $h \in \mathbb{R}[V]$ a quadratic polynomial, i.e.~$h(a, b) = t_1 + t_2 \cdot a + t_3 \cdot b + t_4 \cdot a^2 + t_5 \cdot a \cdot b + t_6 \cdot b^2$. We aim to encode the property that $h$ is a sum-of-squares as a system of quadratic equalities and inequalities. To do so, we first generate all monomials of degree at most $\lfloor d/2 \rfloor = 1,$ which are $1,$ $a$ and $b$. Hence, we let $y = \begin{bmatrix}1 & a & b\end{bmatrix}^T.$ We then generate a lower-triangular matrix $L$ whose every non-zero entry is a new variable: 
	$$
	L = \begin{bmatrix}
	l_1 & 0 & 0\\ 
	l_2 & l_3 & 0 \\ 
	l_4 & l_5  & l_6 
	\end{bmatrix}.
	$$
	We also add the inequalities $l_1 \geq 0, l_3 \geq 0$ and $l_6 \geq 0$ to our system.
	Now, we write the equation $h = y^T L L^T y$ and compute it symbolically:
	$$
	h = \begin{bmatrix}
	1 & a & b
	\end{bmatrix}
	\begin{bmatrix}
	l_1 & 0 & 0\\ 
	l_2 & l_3 & 0 \\ 
	l_4 & l_5  & l_6 
	\end{bmatrix}
	\begin{bmatrix}
	l_1 & l_2 & l_4\\
	0 & l_3 & l_5\\
	0 & 0 & l_6
	\end{bmatrix}
	\begin{bmatrix}
	1\\
	a\\
	b
	\end{bmatrix},
	$$
	which leads to:
	$
	t_1  +  t_2 \cdot a  +  t_3 \cdot b  +  t_4 \cdot a^2  +  t_5 \cdot a \cdot b  +  t_6 \cdot b^2 
	=  l_1^2  +  2 \cdot l_1 \cdot l_2 \cdot a  +  2 \cdot l_1 \cdot l_4 \cdot b  +  (l_2^2 + l_3^2) \cdot a^2  +  (2 \cdot l_2 \cdot l_4 + 2 \cdot l_3 \cdot l_5) \cdot a \cdot b  +  (l_4^2 + l_5^2+l_6^2) \cdot b^2.
	$
	
	Note that both sides of the equation above are polynomials over $\{a, b\},$ hence they are equal iff their corresponding coefficients are equal. So, we get the following quadratic equalities over the $t$-variables and $l$-variables: $t_1 = l_1^2, t_2 = 2 \cdot l_1 \cdot l_2, \ldots, t_6 = l_4^2 + l_5^2 + l_6^2.$
	This concludes the construction of our quadratic system.
\end{example}

\section{Experimental Results} \label{app:exp}

\subsection{The Invariant Synthesized for Our Running Example} \label{app:exp-run}

Table~\ref{tab:simpleex} shows the output of our invariant generation algorithm, i.e.~$\weakInvAlgo$, on the running example of Figure~\ref{fig:running}.

\begin{table}
		\begin{tabular}{|C{2mm}|c|C{5cm}|}
			\hline
			$\loc$ & $\precond(\loc)$ & $\ind(\loc)$\\
			\hline \hline
			$1$ & $\bar{n}=n , n\ge1$ & 
			$0.13 - 0.01 \cdot \bar{n} -0.05 \cdot r -0.07\cdot s -0.24\cdot i + 0.06\cdot n + 0.16\cdot \bar{n}^2 -0.08\cdot  \bar{n}\cdot r + 0.11 \cdot r^2 -0.13 \cdot \bar{n} \cdot s + 0.18\cdot r\cdot s + 0.15\cdot s^2 -0.13\cdot \bar{n}\cdot i + 0.16\cdot  r \cdot i + 0.24 \cdot i\cdot  s+ 0.23\cdot i^2 +0.07 \cdot \bar{n}\cdot  n - 0.52 \cdot r\cdot  n  -0.67\cdot n\cdot s - 0.66 \cdot i\cdot n + 1.10 \cdot n^2 > 0 $\\
			
			\hline
			$2$ & \textbf{true} & $0.09 - 0.01 \cdot \bar{n} - 0.18 \cdot r - 0.30\cdot s + 0.09\cdot i + 0.11\cdot n + 0.03\cdot \bar{n}^2 + 0.01\cdot \bar{n}\cdot r  + 0.13\cdot r^2 -0.03\cdot \bar{n}\cdot s + 0.16\cdot r\cdot s + 0.23\cdot s^2 -0.01\cdot  \bar{n} \cdot i -0.18\cdot r\cdot i  -0.30\cdot i\cdot s + 0.09\cdot i^2 + 0.01\cdot  \bar{n}\cdot  n + 0.11\cdot i\cdot n  > 0$ \\
			
			\hline
			
			$3$ & \textbf{true} & $0.49 + 0.01\cdot \bar{n} + 0.11\cdot r -0.59\cdot s -0.30\cdot i -0.29\cdot n  + 0.13 \cdot r^2 -0.01 \cdot \bar{n}\cdot s -0.35 \cdot r\cdot s + 0.60\cdot s^2 + 0.05\cdot r\cdot i -0.01\cdot i\cdot s + 0.16\cdot i^2 -0.04\cdot r\cdot n -0.06\cdot n \cdot s -0.04\cdot  i\cdot n + 0.18\cdot n^2 > 0$\\
			
			\hline
			
			$4$ & \textbf{true} & $0.20 -0.12 \cdot \bar{n} + 0.01 \cdot r -0.01 \cdot s -0.22\cdot i -0.07 \cdot n + 1.08 \cdot \bar{n}^2 + 0.10 \cdot \bar{n} \cdot r + 0.15\cdot r^2 -0.49 \cdot \bar{n} \cdot s -0.08 \cdot r \cdot s + 0.10 \cdot s^2 -0.65 \cdot \bar{n}\cdot i -0.15\cdot r \cdot i + 0.16 \cdot i \cdot s + 0.14 \cdot i^2 -0.57 \cdot \bar{n} \cdot n -0.11\cdot r \cdot n + 0.13\cdot n \cdot s + 0.24 \cdot i \cdot n + 0.22\cdot n^2 > 0$\\
			
			\hline
			$5$ & \textbf{true} & $0.22 -0.12 \cdot \bar{n} + 0.01 \cdot r -0.02 \cdot s -0.05 \cdot i -0.28 \cdot n+ 1.08\cdot \bar{n}^2 + 0.10\cdot  \bar{n} \cdot r  + 0.15\cdot r^2 -0.49 \cdot \bar{n} \cdot s -0.08 \cdot r \cdot s + 0.10 \cdot s^2 -0.63 \cdot \bar{n}\cdot i -0.13 \cdot r \cdot i  + 0.16 \cdot i \cdot s + 0.14 \cdot i^2 -0.60\cdot \bar{n} \cdot n -0.12 \cdot r \cdot n + 0.13 \cdot n \cdot s + 0.22 \cdot i \cdot n + 0.24 \cdot n^2 > 0$\\
			
			\hline

			$6$ & \textbf{true} & $0.22 -0.12 \cdot \bar{n} + 0.01\cdot r -0.02 \cdot s -0.05 \cdot i -0.28 \cdot n + 1.08 \cdot \bar{n}^2 + 0.10 \cdot \bar{n} \cdot r + 0.15 \cdot r^2 -0.49 \cdot \bar{n} \cdot s -0.08 \cdot r \cdot s + 0.10 \cdot s^2 -0.63 \cdot \bar{n} \cdot i -0.13\cdot r\cdot i + 0.16 \cdot i\cdot s + 0.14 \cdot i^2 -0.60 \cdot \bar{n} \cdot n -0.12 \cdot r \cdot n + 0.13\cdot n \cdot s + 0.22\cdot i \cdot n + 0.24 \cdot n^2 > 0$\\
			
			\hline
			$7$ & \textbf{true} & $0.15 -0.09\cdot \bar{n} -0.12\cdot r + 0.03\cdot s -0.07\cdot i -0.13 \cdot n + 0.23\cdot \bar{n}^2 + 0.48 \cdot \bar{n} \cdot r + 0.35\cdot r^2 -0.31\cdot \bar{n} \cdot s -0.40 \cdot r\cdot s + 0.13 \cdot s^2 + 0.16\cdot \bar{n} \cdot i + 0.13 \cdot r \cdot i -0.09\cdot i \cdot s + 0.06 \cdot i^2 -0.24\cdot \bar{n} \cdot n -0.24 \cdot r\cdot n + 0.18\cdot n \cdot s -0.10 \cdot i \cdot n + 0.14 \cdot n^2 > 0$\\
			
			\hline
			$8$ & \textbf{true} & $0.18 -0.11\cdot \bar{n} + 0.01\cdot r + 0.50 \cdot i -0.79\cdot n + 1.09\cdot \bar{n}^2 + 0.11 \cdot \bar{n} \cdot r + 0.15\cdot r^2 -0.48\cdot \bar{n} \cdot s -0.08\cdot  r \cdot s + 0.10 \cdot s^2 -0.57 \cdot \bar{n} \cdot i -0.09\cdot r \cdot i + 0.16\cdot i \cdot s + 0.18 \cdot i^2 -0.66 \cdot \bar{n} \cdot n -0.16\cdot r \cdot n + 0.12 \cdot n \cdot s + 0.13\cdot i \cdot n + 0.27\cdot n^2 > 0$\\
			\hline
			$9$ & \textbf{true} & $1+0.5 \cdot \bar{n} -r + 0.5 \cdot \bar{n}^2 > 0$\\ \hline
		\end{tabular}
\caption{The inductive invariant generated by $\weakInvAlgo$ for the running example in Figure~\ref{fig:running}}
\label{tab:simpleex}
\end{table}

\clearpage
\subsection{Recursive Examples} \label{app:rec-exp}

We used the following recursive examples as benchmarks. Desired assertions are shown in brackets. Pre-conditions are enclosed in $\#$ signs. In all cases our algorithm synthesizes an inductive invariant that contains the desired assertions. 

\begin{table}[H]
\begin{tabular}{c}

\lstset{tabsize=4}
\lstset{language=prog}
\begin{lstlisting}[mathescape,basicstyle=\small]
$\func{recursive-sum}$($n$) {
	# $n\geq 0$ #
	if $n \leq 0$ then
		return $n$
	else
		$m$ := $n - 1$;
		$s$ := $\func{recursive-sum}$($m$);
		if $\star$ then
			$s$ := $s+n$
		else
			skip
		fi;
		return $s$
	fi
 [$\ret^\func{recursive-sum} < 0.5 \cdot \bar{n}^2 + 0.5 \cdot \bar{n} + 1$]}
\end{lstlisting} 
\end{tabular}
\end{table}

\begin{table}[H]
	\begin{tabular}{c}
\lstset{tabsize=4}
\lstset{language=prog}
\begin{lstlisting}[mathescape,basicstyle=\small]
$\func{recursive-square-sum}$($n$) {
	# $n\geq 0$ #
	if $n \leq 0$ then
		return $n$
	else
		$m$ := $n - 1$;
		$s$ := $\func{recursive-sum}$($m$);
		if $\star$ then
			$s$ := $s+n * n$
		else
			skip
		fi;
		return $s$
	fi
 [$\ret^\func{recursive-square-sum} < 0.34 \cdot \bar{n}^3 + 0.5 \cdot \bar{n}^2 + 0.17 \cdot \bar{n} + 1$]}
\end{lstlisting}
\end{tabular}
\end{table}

\begin{table}[H]
\begin{tabular}{c}
\lstset{tabsize=4}
\lstset{language=prog}
\begin{lstlisting}[mathescape,basicstyle=\small]
$\func{recursive-cube-sum}$($n$) {
	# $n\geq 0$ #
	if $n \leq 0$ then
		return $n$
	else
		$m$ := $n - 1$;
		$s$ := $\func{recursive-sum}$($m$);
		if $\star$ then
			$s$ := $s+n * n * n$
		else
			skip
		fi;
		return $s$
	fi
 [$\ret^\func{recursive-cube-sum} < 0.25 \cdot \bar{n}^2 \cdot (\bar{n}+1)^2 + 1$]}
\end{lstlisting} 
\end{tabular}
\end{table}

\begin{table}[H]
\begin{tabular}{c}
\lstset{tabsize=4}
\lstset{language=prog}
\begin{lstlisting}[mathescape,basicstyle=\small]
$\func{pw2}$($x$) {
	//$\text{computes the largest power of 2 that is } \leq x$
	# $x \geq 1$ #
	if $x \geq 2$ then
		$y$ := $0.5 * x$;
		return $2 * \func{pw2}(y)$
	else
		return $1$
	fi
$[\ret^\func{pw2} \leq \bar{x} \wedge 2 \cdot \ret^\func{pw2} > \bar{x}]$}
\end{lstlisting}
\end{tabular}
\end{table}

\begin{table}
	\begin{tabular}{c}
		\lstset{tabsize=4}
		\lstset{language=prog}
\begin{lstlisting}[mathescape,basicstyle=\small]
$\func{merge-sort}$($s$, $e$) //$\text{sorts and returns number of inversions in} [s..e]$
{
	# $e \geq s$ #
	if $s \geq e$ then
		return $0$
	else
		$i$ := $0.5 * s + 0.5 * e$;
		$j$ := $\lfloor i \rfloor$;
		$i$ := $j+1$;
		$r$ := $\func{merge-sort}(s, j)$;
		$ans$ := $\func{merge-sort}(i, e)$;
		$ans$ := $ans + r$;
		$k$ := $s$;
		
		while $i \leq e$ do
			while $k \leq j$ do
				if $\star$ then //$array[k]\leq array[i]$
					$k$ := $k+1$; 
					skip //$\text{temp.push\_back}$($array[k]$)
				else //$array[k]>array[i]$
					$ans$ := $ans + j - k + 1$; //$\text{add inversions}$
					$i$ := $i+1$;
					skip //$\text{temp.push\_back}$($array[i]$)
				fi
			od;
			skip; //$\text{temp.push\_back}$($array[i]$)
			$i$ := $i+1$
		od;
		
		while $s \leq e$ do
			skip; //$\text{copy from } temp \text{ to } array$	
			$s$ := $s+1$
		od;
		
		return $ans$
	fi
[$\ret^\func{merge-sort} < 0.5 \cdot (\bar{e} - \bar{s}) \cdot (\bar{e} - \bar{s} + 1) + 1$]}
		\end{lstlisting}
	\end{tabular}
\end{table}

\clearpage
\subsection{Example with Two Functions Recursively Calling Each Other} \label{app:twofunc}

Our approach can handle any combination of recursive function calls as long as a polynomial recursive inductive invariant exists. As an example, consider the following program, consisting of two functions $f$ and $g$, which recursively call each other. Our algorithm is able to handle it in 47s using parameters $n=d=\Upsilon=2$. This program leads to a quadratic system of size 5453.

\begin{table}[H]
	\begin{tabular}{c}
		
		\lstset{tabsize=4}
		\lstset{language=prog}
		\begin{lstlisting}[mathescape,basicstyle=\small]
$\func{f}$($n$) {
	# $n\geq 1$ #
	if $n \leq 1$ then
		return $1$
	else
		$x$ := $g(n-1);$
		$x$ := $x + 2 \cdot n - 1;$
		$y$ := $0;$
		while $(y+1)^2 \leq x$ do
			$y$ := $y+1$
		od;
		return $y$
	fi
[$\ret^\func{f} \leq n$]}

$\func{g}$($n$) {
	# $n\geq 1$ #
	return $n \cdot f(n)$
[$\ret^\func{g} \leq n^2$]}
		\end{lstlisting} 
	\end{tabular}
\end{table}

\newpage
\subsection{Example of Synthesizing Polynomial Equality Invariants} \label{app:equality}

To demonstrate that our approach is able to generate invariants including polynomial equalities, we slightly change the program in Figure~\ref{fig:running} to obtain the one in Figure~\ref{fig:equ}. This program precisely computes $1+2+\ldots+n = \frac{n \cdot (n+1)}{2}.$ Let $r$ be the return value of $\func{sum}(n),$ we would like to prove that $r = 0.5 \cdot n^2 + 0.5 \cdot n,$ which is equivalent to $0.5 \cdot n^2 + 0.5 \cdot n - r \geq 0 ~\wedge~ r - 0.5 \cdot n^2 - 0.5 \cdot n \geq 0.$ We can therefore run the sound but incomplete variant of our approach, i.e.~the variant that does not incorporate positivity witnesses, with parameters $n=d=\Upsilon=2.$ Our algorithm successfully synthesizes an inductive invariant that proves the desired equality. The resulting inductive invariant is given in Table~\ref{tab:equ}.

\begin{figure}
	\begin{center}
		\begin{minipage}{0.5\linewidth}
			\lstset{language=prog}
			\lstset{tabsize=4}
			\begin{lstlisting}[aboveskip=0pt,belowskip=0pt,mathescape,basicstyle=\small]
   $\func{sum}$($n$) {
1:	$i$ := $1$;
2:	$s$ := $0$;
3:	while $i \leq n$ do
4:		$s$ := $s+i$;
5:		$i$ := $i+1$
	od;
6:	return s
7: }
			\end{lstlisting}
		\end{minipage}
	\end{center}
	\caption{A summation program}
	\label{fig:equ}
\end{figure}

\begin{table*}
		\begin{tabular}{|C{2mm}|c|C{6.5cm}|C{6.5cm}|}
			\hline
			$\loc$ & $\precond(\loc)$ & $\ind(\loc)$, first inequality & $\ind(\loc)$, second inequality\\
			\hline \hline
			$1$ & $\bar{n}=n , n\ge1$ & 
			$-0.09 - 0.79 \cdot \bar{n} - 0.95 \cdot r - 0.72 \cdot s - 1.12 \cdot i - 0.39 \cdot n + 1.90 \cdot \bar{n}^2 - 0.12 \cdot \bar{n} \cdot r + 1.62 \cdot r^2 - 0.25 \cdot \bar{n} \cdot s - 0.36 \cdot r \cdot s + 1.48 \cdot s^2 - 0.79 \cdot \bar{n} \cdot i - 0.67 \cdot r \cdot i - 0.94 \cdot i \cdot s + 1.57 \cdot i^2 + 0.42 \cdot \bar{n} \cdot n + 0.10 \cdot r \cdot n - 0.15 \cdot n \cdot s - 0.82 \cdot i \cdot n + 2.36 \cdot n^2 \geq 0 $ & 
			
			$0.45 - 0.77 \cdot \bar{n} - 0.85 \cdot r - 0.56 \cdot s - 0.75 \cdot i - 0.47 \cdot n + 1.68 \cdot \bar{n}^2 - 0.37 \cdot \bar{n} \cdot r + 1.45 \cdot r^2 - 0.44 \cdot \bar{n} \cdot s - 0.52 \cdot r \cdot s + 1.42 \cdot s^2 - 0.89 \cdot \bar{n} \cdot i - 0.67 \cdot r \cdot i - 0.89 \cdot i \cdot s + 1.73 \cdot i^2 + 0.09 \cdot \bar{n} \cdot n - 0.20 \cdot r \cdot n - 0.38 \cdot n \cdot s - 0.92 \cdot i \cdot n + 2.07 \cdot n^2 \geq 0$ \\
			
			\hline
			
			$2$ & \textbf{true} & $0.08 - 1.07 \cdot \bar{n} - 1.48 \cdot r - 2.00 \cdot s - 2.42 \cdot i - 0.34 \cdot n + 3.14 \cdot \bar{n}^2 - 0.62 \cdot \bar{n} \cdot r + 3.51 \cdot r^2 - 1.33 \cdot \bar{n} \cdot s - 2.12 \cdot r \cdot s + 2.32 \cdot s^2 + 0.64 \cdot \bar{n} \cdot i + 0.76 \cdot r \cdot i - 1.37 \cdot i \cdot s + 3.40 \cdot i^2 - 0.50 \cdot \bar{n} \cdot n + 0.22 \cdot r \cdot n - 0.68 \cdot n \cdot s + 0.94 \cdot i \cdot n + 3.67 \cdot n^2 \geq 0$ &
			
			$0.76 - 1.37 \cdot \bar{n} - 1.08 \cdot r - 1.64 \cdot s - 1.77 \cdot i - 0.01 \cdot n + 2.25 \cdot \bar{n}^2 - 1.34 \cdot \bar{n} \cdot r + 3.25 \cdot r^2 - 1.20 \cdot \bar{n} \cdot s - 2.27 \cdot r \cdot s + 2.68 \cdot s^2 - 0.04 \cdot \bar{n} \cdot i + 0.59 \cdot r \cdot i - 1.23 \cdot i \cdot s + 2.93 \cdot i^2 - 1.68 \cdot \bar{n} \cdot n - 0.32 \cdot r \cdot n - 0.95 \cdot n \cdot s + 0.31 \cdot i \cdot n + 2.93 \cdot n^2 \geq 0$ \\
			
			\hline
			
			$3$ & \textbf{true} & $0.94 - 0.37 \cdot \bar{n} - 0.07 \cdot r + 0.95 \cdot s + 0.38 \cdot i - 0.09 \cdot n + 0.07 \cdot \bar{n}^2 + 0.03 \cdot \bar{n} \cdot r - 0.19 \cdot \bar{n} \cdot s - 0.04 \cdot r \cdot s + 0.28 \cdot s^2 - 0.07 \cdot \bar{n} \cdot i - 0.01 \cdot r \cdot i + 0.27 \cdot i \cdot s + 0.11 \cdot i^2 + 0.01 \cdot \bar{n} \cdot n + 0.02 \cdot n \cdot s + 0.06 \cdot i \cdot n + 0.18 \cdot n^2 \geq 0$ &
			
			$3.69 + 3.15 \cdot \bar{n} - 0.09 \cdot r - 3.37 \cdot s - 0.06 \cdot i - 0.19 \cdot n + 4.87 \cdot \bar{n}^2 - 0.01 \cdot \bar{n} \cdot r + 0.01 \cdot r^2 - 1.44 \cdot \bar{n} \cdot s - 0.09 \cdot r \cdot s + 2.89 \cdot s^2 + 0.01 \cdot \bar{n} \cdot i - 0.05 \cdot i \cdot s + 0.02 \cdot i^2 + 0.02 \cdot \bar{n} \cdot n + 0.02 \cdot r \cdot n - 0.25 \cdot n \cdot s + 0.03 \cdot n^2 \geq 0$\\
			
			\hline
			
			$4$ & \textbf{true} & $0.89 - 0.31 \cdot \bar{n} - 0.06 \cdot r + 0.91 \cdot s + 0.39 \cdot i - 0.26 \cdot n + 0.06 \cdot \bar{n}^2 + 0.02 \cdot \bar{n} \cdot r - 0.16 \cdot \bar{n} \cdot s - 0.03 \cdot r \cdot s + 0.26 \cdot s^2 - 0.06 \cdot \bar{n} \cdot i - 0.01 \cdot r \cdot i + 0.27 \cdot i \cdot s + 0.11 \cdot i^2 + 0.01 \cdot \bar{n} \cdot n - 0.09 \cdot n \cdot s - 0.02 \cdot i \cdot n + 0.23 \cdot n^2 \geq 0$ &
			
			$0.34 + 0.47 \cdot \bar{n} - 0.55 \cdot s - 0.56 \cdot i + 0.01 \cdot n + 0.82 \cdot \bar{n}^2 - 0.01 \cdot \bar{n} \cdot r - 0.16 \cdot \bar{n} \cdot s + 0.25 \cdot s^2 - 0.15 \cdot \bar{n} \cdot i + 0.51 \cdot i \cdot s + 0.26 \cdot i^2 - 0.01 \cdot \bar{n} \cdot n - 0.01 \cdot n \cdot s - 0.01 \cdot i \cdot n \geq 0$\\
			
			\hline
			$5$ & \textbf{true} & $2.83 - 1.01 \cdot \bar{n} - 0.20 \cdot r + 2.86 \cdot s - 1.62 \cdot i - 0.86 \cdot n + 0.26 \cdot \bar{n}^2 + 0.10 \cdot \bar{n} \cdot r + 0.01 \cdot r^2 - 0.52 \cdot \bar{n} \cdot s - 0.11 \cdot r \cdot s + 0.83 \cdot s^2 + 0.30 \cdot \bar{n} \cdot i + 0.05 \cdot r \cdot i - 0.80 \cdot i \cdot s + 0.32 \cdot i^2 + 0.04 \cdot \bar{n} \cdot n + 0.01 \cdot r \cdot n - 0.27 \cdot n \cdot s + 0.21 \cdot i \cdot n + 0.73 \cdot n^2 \geq 0$ &
			
			$2.03 + 2.31 \cdot \bar{n} - 0.04 \cdot r - 2.82 \cdot s - 0.01 \cdot i - 0.02 \cdot n + 4.08 \cdot \bar{n}^2 - 0.03 \cdot \bar{n} \cdot r - 0.79 \cdot \bar{n} \cdot s - 0.01 \cdot r \cdot s + 1.32 \cdot s^2 + 0.03 \cdot \bar{n} \cdot i + 0.02 \cdot i^2 - 0.06 \cdot \bar{n} \cdot n + 0.01 \cdot r \cdot n - 0.06 \cdot n \cdot s + 0.01 \cdot n^2 \geq 0$\\

			\hline
			$6$ & \textbf{true} & $20.37 - 8.79 \cdot \bar{n} - 0.05 \cdot r + 25.48 \cdot s - 0.19 \cdot i - 2.16 \cdot n + 4.25 \cdot \bar{n}^2 - 0.25 \cdot \bar{n} \cdot r + 2.89 \cdot r^2 - 7.49 \cdot \bar{n} \cdot s + 2.00 \cdot r \cdot s + 25.53 \cdot s^2 - 1.08 \cdot \bar{n} \cdot i - 0.97 \cdot r \cdot i - 0.64 \cdot i \cdot s + 2.83 \cdot i^2 - 1.79 \cdot \bar{n} \cdot n - 2.21 \cdot r \cdot n - 4.19 \cdot n \cdot s - 0.86 \cdot i \cdot n + 4.49 \cdot n^2 \geq 0$ &
			
			$68.26 + 60.81 \cdot \bar{n} - 2.84 \cdot r - 68.34 \cdot s - 2.64 \cdot i - 3.44 \cdot n + 100.00 \cdot \bar{n}^2 + 0.76 \cdot \bar{n} \cdot r + 2.64 \cdot r^2 - 33.53 \cdot \bar{n} \cdot s - 2.63 \cdot r \cdot s + 57.63 \cdot s^2 + 0.71 \cdot \bar{n} \cdot i - 1.22 \cdot r \cdot i - 2.45 \cdot i \cdot s + 2.36 \cdot i^2 + 0.91 \cdot \bar{n} \cdot n - 1.74 \cdot r \cdot n - 3.16 \cdot n \cdot s - 1.82 \cdot i \cdot n + 2.19 \cdot n^2 \geq 0$\\
			
			\hline
			$7$ & \textbf{true} & $0.50 \cdot n^2 + 0.50 \cdot n - r \geq 0$ & 
			
			$r - 0.50 \cdot n^2 - 0.50 \cdot n \geq 0$\\
			\hline
		\end{tabular}
\caption{The inductive invariant generated by $\weakInvAlgo$ for the summation program in Figure~\ref{fig:equ}}
\label{tab:equ}
\end{table*}

\clearpage
\subsection{Continued Fraction Example} \label{app:frac}

Given that our approach has semi-completeness guarantees (over bounded reals), it is no surprise that it can generate desired polynomial invariants for inputs which no previous incomplete approach could handle. We now present a classical example of a program that approximates $\sqrt{2}$ using its continued fraction representation. Our implementation generates required invariants of degree $5$, which is beyond the reach of all previous methods in Table~\ref{tab:comp}. Specifically, we manually tried all the incomplete approaches in Table~\ref{tab:comp} over this example. They are either not applicable or fail to synthesize the desired invariant. However, some of them synthesize other invariants for the same program. 

We first review some well-known facts about continued fractions. A continued fraction is an expression of the following form:
$$
x = a_0 + \frac{1}{a_1 + \frac{1}{a_2 + \ldots}}
$$
in which the $a_i$'s are natural numbers. For brevity, we denote this fraction as $x = [a_0; a_1, a_2, \ldots].$ Note that the continued fraction representation might be finite (in case of rational numbers) or infinite (in case of irrationals). Specifically, it is easy to verify that $\sqrt{2} = [1; 2,2,2, \ldots].$ In any case, we define $x_n := [a_0; a_1, \ldots, a_n]$ and call it the $n$-th convergent of $x$. A standard way for approximating irrational numbers is to evaluate the convergents of their continued fraction representation. We consider a program that approximates $\sqrt{2}$ using this technique.

The following well-known lemma provides some properties of the convergents and a simple algorithm for computing them:
\begin{lemma}[\cite{contfrac}] \label{lemm:nechaev}
	Let $x = [a_0; a_1, a_2, \ldots].$ We define two sequences $\langle P_n \rangle_{n=-2}^\infty$ and $\langle Q_n \rangle_{n=-2}^\infty$ as follows:
	$$
	\begin{matrix}	P_n = a_n \cdot P_{n-1} + P_{n-2} & & & P_{-1} = 1 & P_{-2} = 0\\
	Q_n = a_n \cdot Q_{n-1} + Q_{n-2} & & & Q_{-1} = 0 & Q_{-2} = 1
	\end{matrix}.
	$$
	The following properties hold for all $0 \leq n < \infty$:
	\begin{enumerate}[(i)]	
		\item $x_n = \frac{P_n}{Q_n}.$
		\item $\vert x - x_n \vert \leq \vert x - x_{n+1} \vert.$
		\item $\vert x - x_n \vert \leq \frac{1}{Q_n \cdot Q_{n+1}}.$
		\item $Q_n \cdot P_{n-1} - Q_{n-1} \cdot P_n = (-1)^n.$
	\end{enumerate} \qed
\end{lemma}

Consider the program in Figure~\ref{prog:contfrac}. This program computes the values of $P_n$ and $Q_n$ for every $n \geq 0$. We use the variable $p_0$ to save values of $P_{2n}$, i.e.~even-indexed values of the sequence $P$, and $p_1$ to save $P_{2n+1}.$ The variables $q_0$ and $q_1$ are used in a similar manner. We can encode properties (ii) to (iv) of Lemma~\ref{lemm:nechaev} as partial invariants for the program in Figure~\ref{prog:contfrac} and check if our approach can synthesize inductive invariants that prove them.

\begin{figure}[H]
	\begin{tabular}{c}
		\lstset{tabsize=4}
		\lstset{language=prog}
		\begin{lstlisting}[mathescape,basicstyle=\small]
  $\func{continued-fraction}$()
  {
1:	$p_0$ := $1$;
2:	$p_1$ := $3$;
3:	$q_0$ := $1$;
4:	$q_1$ := $2$;
	
5:	while $1 \geq 0$ do
6:		$p_0$ := $2 \cdot p_1 + p_0$;
7:		$q_0$ := $2 \cdot q_1 + q_0$;
8:		skip; //$x_n$ := $p_0 / q_0$;
		
9:		$p_1$ := $2 \cdot p_0 + p_1$;
10:		$q_1$ := $2 \cdot q_0 + q_1$;
11:		skip //$x_n$ := $p_1 / q_1$
	od
}
		\end{lstlisting}
	\end{tabular}
	\caption{A program that approximates $\sqrt{2}$ by computing the two sequences defined in Lemma~\ref{lemm:nechaev}.}
	\label{prog:contfrac}
\end{figure}

 Specifically, letting $x = \sqrt{2}$ we want the algorithm to prove the following partial invariants:
\begin{itemize}
	\item[(ii)] $\vert x - x_n \vert \leq \vert x - x_{n+1} \vert$ can be rewritten as follows:
	\begin{gather*}
		(x - x_n)^2 \leq (x - x_{n+1})^2\\
		\left(x - \frac{P_n}{Q_n}\right)^2 \leq \left(x - \frac{P_{n+1}}{Q_{n+1}}\right)^2\\
		(Q_n \cdot Q_{n+1} \cdot x - P_n \cdot Q_{n+1})^2 \leq (Q_n \cdot Q_{n+1} \cdot x - P_{n+1} \cdot Q_n)^2
	\end{gather*}
	which is a polynomial inequality. Hence, we aim to find an inductive invariant that contains the following inequalities at lines 8 and 11, respectively:
	$$
	(q_0 \cdot q_1 \cdot \sqrt{2} - p_1 \cdot q_0)^2 \leq (q_0 \cdot q_1 \cdot \sqrt{2} - p_0 \cdot q_1)^2,
	$$
	$$
	(q_0 \cdot q_1 \cdot \sqrt{2} - p_0 \cdot q_1)^2 \leq (q_0 \cdot q_1 \cdot \sqrt{2} - p_1 \cdot q_0)^2.
	$$
	To model $\sqrt{2},$ we consider a new program variable $x$ whose value is always $\sqrt{2}.$ We enforce this by adding $x^2 \leq 2 \wedge x^2 \geq 2 \wedge x \geq 0$ to every pre-condition. 
	
	\item[(iii)] Similar to the previous case, this property can be rewritten as a polynomial inequality as follows:
	\begin{gather*}
		\vert x - x_n \vert \leq \frac{1}{Q_n \cdot Q_{n+1}}\\				
		\vert x - x_n \vert \cdot Q_n \cdot Q_{n+1} \leq 1\\
		(x-x_n)^2 \cdot Q_n^2 \cdot Q_{n+1}^2 \leq 1\\
		\left(x - \frac{P_n}{Q_n} \right)^2 \cdot Q_n^2 \cdot Q_{n+1}^2 \leq 1\\
		(Q_n \cdot x - P_n)^2 \cdot Q_{n+1}^2 \leq 1		
	\end{gather*}
	Therefore, we aim to find an inductive invariant that contains the following inequalities at lines 8 and 11, respectively:
	$$
	(q_1 \cdot \sqrt{2} - p_1)^2 \cdot q_0^2 \leq 1,
	$$
	$$
	(q_0 \cdot \sqrt{2} - p_0)^2 \cdot q_1^2 \leq 1.
	$$
	
	\item[(iv)] This property is already in polynomial form wrt our program variables. Therefore, it corresponds to the following equality at lines 8 and 11:
	$$
		q_0 \cdot p_1 - q_1 \cdot p_0 = 1.
	$$	
\end{itemize} 

\paragraph{Execution Results} We ran our approach on the program of Figure~\ref{prog:contfrac} with the goal of finding an inductive invariant containing the partial invariants listed above. Note that the partial invariants in (ii) are of degree $5$ and hence we set $d = \Upsilon = 5.$ Moreover, we set $n=5$, i.e.~we generate $5$ polynomial inequalities at each program point. Using these parameters, our approach was able to successfully generate the desired inductive invariant in 48m. To the best of our knowledge, no previous approach for polynomial invariant generation can handle this example.

\begin{remark}
Note that the choice of $x=\sqrt{2}$ in the example above was arbitrary. One can replace $\sqrt{2}$ with any other real number with a periodic continued fraction representation, thus obtaining a family of programs whose desired partial invariants (as in Lemma~\ref{lemm:nechaev}) can be automatically proven by our approach, but not by any of the previous approaches. It is well-known that the set of real numbers with periodic continued fraction representation is the same as the set of quadratic irrationals, i.e.~irrational roots of quadratic equations with integer coefficients~\cite{periodic}. For example, this set contains $\sqrt{n}$ for every non-square $n \in \mathbb{N}.$
\end{remark}

\end{document}